\documentclass[journal, onecolumn ]{IEEEtran} 
\usepackage[utf8]{inputenc} 
\usepackage{amsfonts}      
\usepackage{xcolor}  
\usepackage{times}
\usepackage[a4paper, total={6.5in, 9in}]{geometry}
\usepackage{hyperref}
\usepackage{mathtools}
\usepackage{enumerate}
\usepackage{float}
\usepackage{bbm}
\usepackage{caption}
\usepackage[english]{babel}
\usepackage{dsfont}	
\usepackage{tikz}
\usepackage{dirtytalk}
\usepackage{enumitem,color,bbm}
\usepackage{amssymb,amsmath,amsthm}
\usepackage{setspace}
\usetikzlibrary{arrows,arrows.meta}
\usepackage{algorithm}
\usepackage{algpseudocode}

\def\dtv{d_{\mathsf{TV}}}
\def\Bern{\mathsf{Bern}}
\newcommand{\veps}{\varepsilon}
\newcommand{\eps}{\epsilon}
\newtheorem{theorem}{Theorem}
\newtheorem{lemma}{Lemma}

\DeclareMathOperator{\polylog}{polylog}
\DeclareMathOperator{\E}{\mathbb{E}}
\DeclareMathOperator{\Pe}{\mathsf{P_e}}

\DeclareMathOperator{\Var}{\mathsf{Var}}

\title{Memory Complexity of Estimating Entropy and Mutual Information}

\author{Tomer~Berg, Or~Ordentlich and Ofer~Shayevitz}

\begin{document}

\maketitle

\begin{abstract}
  We observe an infinite sequence of independent identically distributed random variables $X_1,X_2,\ldots$ drawn from an unknown distribution $p$ over $[n]$, and our goal is to estimate the entropy $H(p)=-\E[\log p(X)]$ within an $\veps$-additive error. To that end, at each time point we are allowed to update a finite-state machine with $S$ states, using a possibly randomized but time-invariant rule, where each state of the machine is assigned an entropy estimate. Our goal is to characterize the minimax memory complexity $S^*$ of this problem, which is the minimal number of states for which the estimation task is feasible with probability at least $1-\delta$ asymptotically, uniformly in $p$. Specifically, we show that there exist universal constants $C_1$ and $C_2$ such that $ S^* \leq C_1\cdot\frac{n (\log n)^4}{\veps^2\delta}$ for $\veps$ not too small, and $S^* \geq C_2 \cdot \max \{n, \frac{\log n}{\veps}\}$ for $\veps$ not too large. The upper bound is proved using approximate counting to estimate the logarithm of $p$, and a finite memory bias estimation machine to estimate the expectation operation. The lower bound is proved via a reduction of entropy estimation to uniformity testing. We also apply these results to derive bounds on the memory complexity of mutual information estimation.
\end{abstract}

\section{Introduction}
The problem of inferring properties of an underlying distribution given sample access is called \textit{statistical property estimation}. A typical setup is as follows: given independent samples $X_1, \ldots ,X_n$ from an unknown distribution $p$, the objective is to estimate a property $g(p)$ (e.g., entropy, support size, $L_p$ norm, etc.) under some resource limitation. A prominent example of such a limitation is the amount of available samples, and this limitation gives rise to the notion of sample complexity, namely the minimal number of samples one needs to see in order to estimate $g(p)$ with some given accuracy. Many real–world machine learning and data analysis tasks are limited by insufficient samples, and the challenge of inferring properties of a distribution given a small sample size is encountered in a variety of settings, including text data, customer data, and the study of genetic mutations across a population. The sample complexity of property estimation and, specifically, of entropy estimation, have therefore received much attention in the literature (see Section~\ref{sec:related} for details). 

However, in many contemporary settings, collecting enough samples for accurate estimation is less of a problem, and the bottleneck shifts to the computational resources available for the task and, in particular, the available memory size. In this work, we therefore focus on the problem of estimation under memory constraints, and, in particular, entropy estimation. In order to isolate the effect that finite memory has on the fundamental limits of the problem, we let the number of samples we process be arbitrarily large. 

Formally, the problem is defined as follows. Let $\Delta_n$ be the collection of all distributions over $ [n] $. The Shannon entropy of $p \in \Delta_n$ is $H (p) = -\sum_{x\in [n]}p(x)\log p(x)$. Given independent samples $X_1,X_2,\ldots$ from an unknown $p \in \Delta_n$, we would like to accurately estimate $H(p)$ using limited memory. To that end, an \textit{$S$-state entropy estimator} is a finite-state machine with $S$ states, defined by two functions: The (possibly randomized) memory update function $f:[S] \times [n] \rightarrow [S]$, describing the transition between states as a function of an input sample, and the entropy estimate function $\hat{H}:[S]\rightarrow [0,\log n]$, assigning an entropy estimate to each state. Letting $M_t$ denote the state of the memory at time $t$, this finite-state machine evolves according to the rule: 
\begin{align}
M_0&=s_{\text{init}},\label{eq:init} \\M_t&=f(M_{t-1},X_t)\in [S],\label{eq:evolution}
\end{align}
for some predetermined initial state $s_{\text{init}} \in [S]$. If the machine is stopped at time $t$, it outputs the estimation $\hat{H}(M_t)$. We restrict the discussion to time-invariant memory update function $f$, since storing the time index necessarily incurs a memory cost, and, furthermore, since the number of samples is unbounded, simply storing the code generating a time-varying algorithm may require unbounded memory. We say that an $\eps$-error occurred at time $t$ if our estimate $\hat{H}(M_t)$ is $\veps$-far from the correct entropy. Our figure of merit for the estimator is taken to be its worst-case asymptotic $\eps$-error probability:
\begin{align}
	 \Pe (f,\hat{H},\eps) &= \underset{p\in \Delta_n}{\sup}\limsup_{t\rightarrow\infty} \Pr \left(|\hat{H}(M_t)-H(p)|>\veps \right)\label{eq:pefd}. 
\end{align}
We are interested in the \emph{minimax memory complexity} $S^*(n,\veps,\delta)$, defined as the smallest integer $S$ for which there exist $(f,\hat{H})$ such that $\Pe (f,\hat{H},\eps)\leq \delta$.

Our main result is an upper bound on $S^*(n,\veps,\delta)$, which shows that $\log \frac{n}{\veps^2}+o\left(\log n\right)$ bits suffice for entropy estimation when $\veps>10^{-5}$, thus improving upon the best known upper bounds thus far (\cite{acharya2019estimating,aliakbarpour2022estimation}). While our focus here is on minimizing the memory complexity of the problem in the limit of infinite number of available samples, we further show that the estimation algorithm attaining this memory complexity upper bound only requires $\tilde{O}(n^c)$ samples, for any $c>1$.\footnote{The $\tilde{O}$ suppresses poly-logarithmic terms.} Thus, in entropy estimation one can achieve almost optimal sample complexity and memory complexity, simultaneously. 
Our proposed algorithm approximates the logarithm of $p(x)$, for a given $x\in[n]$, using a \textit{Morris counter}~\cite{morris1978counting}. The inherent structure of the Morris counter is particularly suited for constructing a nearly-unbiased estimator for $\log p(x)$, making it a natural choice for memory efficient entropy estimation. In order to compute the mean of these estimators, $\E[\widehat{\log p(X)}]$, in a memory efficient manner, a finite-memory bias estimation machine (e.g.,~\cite{leighton1986estimating,berg2021deterministic}) is leveraged for simulating the expectation operator. The performance of a scheme based on this high-level idea is analyzed, and yields the following upper bound on the memory complexity:
\begin{theorem}\label{thm:upper_bound}
For any $c>1$, $\beta>0$, $0<\delta<1$ and $\veps=10^{-5}+\beta+\psi_c(n)$, we have
\begin{align}
    S^*(n,\veps,\delta) \leq n\left(\frac{8(c\log n +2)^4}{\beta^2\delta}+4(c\log n +2)^2\right),\label{eq:upb} 
\end{align}
where 
\begin{align}
  \psi_c(n)&=(e+1)n^{-(c-1)+v_n(1)}+\min\{1,C\cdot  n^{-\frac{c-1}{2}+v_n(1/2)}\}+n^{-c}\cdot\frac{ 100(c\log n+2)}{(1-0.5n^{-c})^2} \nonumber\\&=O\left(2^{\sqrt{\log n}}\cdot n^{-\frac{c-1}{2}}\right),
\end{align}
and we set $C= 2(e+1) 10^8$ and $v_n(\alpha)\triangleq\sqrt{\frac{2c\alpha^3}{\log n}}+\frac{\alpha}{\log n}$.\\
Moreover, there is an algorithm that attains~\eqref{eq:upb} when the number of samples is $\Omega \left(\frac{n^c\cdot \mathop{\mathrm{poly}}(\log n)}{\delta}\cdot \mathop{\mathrm{poly}}(\log (1/\delta))\right)$, and returns an estimation of $H(p)$ within an $\veps$-additive error with probability at least $1-3\delta$.
\end{theorem}

Note that the additive term $\psi_c(n)$ only becomes negligible when $10^8$ is much smaller than $n^{-\frac{c-1}{2}}$, thus the regime in which our results are significant is the asymptotic regime. Furthermore, while $\psi_c(n)$ vanishes for large $n$, our bound is always limited to $\veps>10^{-5}$. This small bias is due to inherent properties of the Morris counter, on which we elaborate in Section~\ref{sec:pre}. As we are more interested in the case where the entropy grows with the alphabet size $n$, the limitation of the attainable additive error to  values above $10^{-5}$ is typically a very moderate one for the sizes of $n$ we consider. While attaining good sample complexity is not the main focus of our work, we also note that if $n$ is large and $\veps$ not too small, one can choose $c$ arbitrarily close to $1$, resulting in an algorithm whose sample complexity has similar dependence on $n$ as those of the limited-memory entropy estimation algorithms proposed in~\cite{acharya2019estimating} and~\cite{aliakbarpour2022estimation}, while requiring less memory states. This result might be of practical interest for applications in which memory is a scarcer resource than samples, e.g., a limited memory high-speed router that leverages entropy estimation to monitor IP network traffic~\cite{chakrabarti2006estimating}. Finally, we note that initial simulation results are more optimistic than the prediction of the theorem. Specifically, for a uniformly distributed input over $n=1000$  with parameters $c=1.5,\beta=0.1,\delta=0.1$, the sample complexity of the algorithm for these parameters, as prescribed by lemma~\ref{lem:sample_comp} is $L\approx 4\cdot 10^{14}$. However, by running the entropy machine for $t=10^{11}$ samples, we were able to obtain an additive error of about $ 0.15$, much smaller than the infinite sample additive error predicted by Theorem~\ref{thm:upper_bound}, which is approximately
$\beta +\min \{1,C\cdot n^{-1/2}\}+n^{-c}\cdot 100(c\log n+2)\approx 1.18$. It seems that $\psi_c(n)$, while negligible for large $n$, is still marginally loose.

Furthermore, we derive two lower bounds on the memory complexity. The first lower bound shows that when $H(p)$ is close to $\log{n}$, the memory complexity cannot be too small. This bound is obtained via a reduction of entropy estimation to uniformity testing, by noting that thresholding the output of a good entropy estimation machine around $\log{n}$ can be used to decide whether $p$ is close to the uniform distribution or not. The bound then follows from the $\Omega(n)$ lower bound of~\cite{berg2022memory}  on uniformity testing. The second lower bound follows from the observation that, if the number of states is too small, there must be some value of the entropy at distance greater than $\veps$ from all estimate value hence, for this value of the entropy, our entropy estimator will be $\veps$-far from the real value with probability $1$. Combining these lower bounds yields the following.
\begin{theorem}\label{thm:lower_bound}
For any $\veps>0$, we have
	\begin{align}
	S^*(n,\veps,\delta) \geq  \frac{
 \log n}{2\veps} . 
	\end{align}
 Furthermore, if $\veps<\frac{1}{4\ln 2}$, then
 \begin{align}
	S^*(n,\veps,\delta) \geq   n(1-2\sqrt{\veps\ln 2}). 
\end{align}
\end{theorem}
One of several open problems posed by the authors of~\cite{acharya2019estimating} is to prove a lower bound on the space requirement of a sample optimal algorithm for entropy estimation. Theorem~\ref{thm:lower_bound} answers this question by giving a lower bound on the memory size needed when the number of samples is infinite, which clearly also holds for any finite number of samples.
In the concluding section of the paper, we extend our results to the mutual information estimation problem. Let $(X,Y)\sim p_{XY}$, where $p_{XY}$ is an unknown distribution over $ [n]\times [m]$ such that the marginal distribution of $X$ is $p_X$ and the marginal distribution of $Y$ is $p_Y$. The mutual information between $X$ and $Y$ is given as $I(X;Y)=H(X)+H(Y)-H(X,Y)$. We derive the following bounds on the memory complexity of mutual information estimation, namely the minimal number of states needed to estimate $I(X;Y)$ with additive error at most $\veps$ with probability of at least $1-\delta$, which we denote as $S_{\text{MI}}^*(n,m,\veps,\delta)$.
\begin{theorem}\label{thm:mutual} For any $c>1$, $\beta>0$ and $\veps=3\cdot 10^{-5}+\beta+O\left(\min \left\{2^{\sqrt{\log n}}\cdot n^{-\frac{1}{2}\cdot (c-1)}, 2^{\sqrt{\log m}}\cdot m^{-\frac{1}{2}\cdot (c-1)}\right\}\right)$,
\begin{align}
	S_{\text{MI}}^*(n,m,\veps,\delta) \leq nm\left(\frac{288\cdot (c\log nm+2)^6}{\beta^2\delta}+16(c\log nm+2)^4\right)
\end{align}
For $\veps<\frac{1}{12\ln 2}$, and if $\frac{n}{\log^3 n}=\Omega(\log^7 m)$ and $\frac{m}{\log^3 m}=\Omega(\log^7 n)$ both hold, then
\begin{align}
	 S_{\text{MI}}^*(n,m,\veps,\delta) =\Omega\left(\frac{n\cdot m}{\log^3 n\cdot \log^3 m}\right).
\end{align}
\end{theorem}

\section{Related work}\label{subsec:related_work}\label{sec:related}

The study of estimation under memory constraints has received far less attention then the sample complexity of statistical estimation. References~\cite{cover1969hypothesis},~\cite{hellman1970learning} studied this setting for hypothesis testing with finite memory, and~\cite{samaniego1973estimating},~\cite{leighton1986estimating} have studied estimating the bias of a coin using a finite state machine. It has then been largely abandoned, but recently there has been a revived
interest in space-sample trade-offs in statistical estimation, and many works have addressed different aspects of the learning under memory constraints problem over the last few years. See, e.g.,~\cite{sd15,svw16,raz18,ds18,jain2018effective,dks19,ssv19,garg2021memory,pensia2022communication} for a non exhaustive list of recent works.

The problem of estimating the entropy with limited independent samples from the distribution has a long history. It was originally addressed by~\cite{basharin1959statistical}, who suggested the simple and natural empirical plug-in estimator. This estimator outputs the entropy of the empirical distribution of the samples, and its sample complexity ~\cite{antos2001convergence} is $\Theta\left(\frac{n}{\veps }+\frac{\log ^2n}{\veps^2}\right)$.~\cite{antos2001convergence} showed that the plug-in estimator is always consistent, and the resulting sample complexity was shown to be linear in $n$. In the last two decades, many efforts were made to improve the bounds on the sample complexity. Paninski~\cite{paninski2003estimation,paninski2004estimating} was the first to prove that it is possible to consistently estimate
the entropy using sublinear sample size. While the scaling of the minimal sample size of consistent estimation was shown to be $\frac{n}{\log n}$ in the seminal results of~\cite{valiant2010clt,valiant2011estimating}, the optimal dependence of the sample size on both $n$ and $\veps$ was not completely resolved until recently. In particular, $\Omega\left(\frac{n}{\veps\log n}\right)$ samples were shown to be necessary, and the best upper bound on the sample complexity was  relied on an estimator based on linear programming that can achieve an additive error $\veps$ using $O\left(\frac{n}{\veps^2\log n}\right)$ samples~\cite{valiant2013estimating}. This gap was partially amended in~\cite{valiant2011power} by a different estimator, which requires $O\left(\frac{n}{\veps\log n}\right)$ samples but is only valid when $\veps$ is not too small. The sharp sample complexity was shown by~\cite{jiao2015minimax,wu2016minimax} to indeed be   
\begin{align}
  \Theta\left(\frac{n}{\veps \log n}+\frac{\log ^2n}{\veps^2}\right).  
\end{align}

The space-complexity (which is the minimal memory in bits needed for the algorithm) of estimating the entropy of the empirical distribution of the data
stream is well-studied for worst-case data streams of a given length, see~\cite{lall2006data},~\cite{chakrabarti2006estimating},~\cite{guha2009sublinear}. Reference~\cite{chien2010space} addressed the problem of deciding if the entropy
of a distribution is above or beyond than some predefined threshold, using algorithms with limited memory. The trade-off between sample complexity and space/communication complexity for the entropy estimation of a distribution is the subject of a more recent line of work. The earliest work on the subject is~\cite{acharya2019estimating}, 
where the authors constructed an algorithm which is guaranteed to work with $O(n/\veps^3 \cdot \polylog (1/\veps))$ samples and any memory size $b\geq 20 \log \left(\frac{n}{\veps}\right)$ bits (which corresponds to $O(n^{20}/\veps^{20})$ memory states in our setup). Their upper bound on the sample complexity was later improved by~\cite{aliakbarpour2022estimation} to $O(n/\veps^2 \cdot \polylog (1/\veps))$ with space complexity of $O\left(\log \left(\frac{n}{\veps}\right)\right)$ bits. We note that in both works above the constant in the space complexity upper bound can be reduced from $20$ to $5$ by a careful analysis. The work of~\cite{acharya2019estimating} is based on an empirical estimator. Consider the following 
simple approach: we draw $N$ samples from the distribution to approximate $\hat{p}_x$ and then take the average of $\log(1/\hat{p}_x)$ over $R$ iterations. This approach gives the desired memory bound, but uses too many samples. To improve the sample complexity, the authors suggest to partition $[0,1]$ into $T$
disjoint intervals, and perform the simple approach above separately for probabilities within each interval. The essence of the algorithm is that when $p_X$ is large inside the interval, fewer samples are needed to estimate $p_X$ (small $N$), and if $p_X$ is small inside the interval, fewer iterations are needed (small $R$). The algorithm of~\cite{aliakbarpour2022estimation} is based on the observation that $Y$, the number of additional draws needed to see $x$ exactly $t$ more times (where $t$ is an algorithm parameter) is a negative binomial random variable, $Y\sim\mathsf{NB}(t,p_x)$. As $\E(Y)=t/p_x$, taking $\log(Y/t)$ as the estimate of $\log(1/p_x)$, and adding a few more samples to correct the bias, allows the authors to improve the sample complexity by a factor of $1/\veps$. Both of the approaches above can be referred to as “natural” approaches, i.e., realizing algorithms that work well for the unconstrained setup in a memory limited framework. It has been pointed out recently in~\cite{berg2024statistical} that this approach for obtaining upper bounds in a memory limited scheme can be strictly suboptimal in the large sample regime. As an example, consider the natural statistic for estimating the parameter of a coin, which is counting the number of $1$’s in a stream. This results in a quadratic risk of $O(1/\sqrt{S})$, as in order to count the number of $1$’s in a stream of length $k$ we must keep a clock that counts
to $k$, thus overall the number of states used is $S = O(k^2)$. However, it is known from the works of~\cite{leighton1986estimating,berg2021deterministic} that the best achievable quadratic risk is $O(1/S)$, and it can be achieved by randomized or deterministic constructions. Unlike the works of~\cite{acharya2019estimating,aliakbarpour2022estimation}, where the authors try to estimate $p_x$ by drawing some related r.v. (Binomial or Negative-Binomial) and then taking its normalized logarithm, the algorithm we purpose directly estimates $\log p_x$ by leveraging properties of Morris counters. This allows our algorithm to save memory, yet it comes at a price of an added periodic term, inherent to Morris counters, that can be only bounded by $10^{-5}$. Whether this term is a real bottleneck or an artifact of our algorithm is a topic for further research. 


\section{Preliminaries}\label{sec:pre}
In this section, we introduce mathematical notations and some relevant background for the paper.
\subsection{Notation}
We write $[n]$ to denote the set $\{1,\ldots,n\}$, and consider discrete distributions over $[n]$. We use the notation $p_i$ to denote the probability of element $i$ in distribution $p$. When $X$ is a random variable on $[n]$, $p_X$ denotes the random variable obtained by evaluating $p$ in location $X$. The entropy of $p$ is defined as $H (p) = -\sum_{x\in [n]}p_x\log p_x=\E_{X\sim p} (-\log p_X)$, where $ H (p)=0$ for a single mass distribution and $ H (p)=\log n$ a uniform distribution over $[n]$. The total variation distance between distributions $p$ and $q$ is defined as half their $\ell^1$ distance, i.e., $\dtv(p,q)=\frac{1}{2}||p-q||_1=\frac{1}{2}\sum_{i=1}^n |p_i-q_i|$, and their KL (Kullback–Leibler) divergence is defined as $D_{\text{KL}}(p||q)=\sum_{i=1}^n p_i\log \frac{p_i}{q_i}$. Logarithms are taken to base 2.

\subsection{Morris Counter}
Suppose one wishes to implement a counter that counts up to $m$. Maintaining this counter exactly can be accomplished using $\log m$ bits. In the first example of a non-trivial streaming algorithm, Morris gave a randomized “approximate counter”, which allows one to retrieve a constant multiplicative approximation to $m$ with high probability using  only $O(\log \log m)$ bits (see~\cite{morris1978counting}). The Morris Counter was later analyzed in more detail by Flajolet~\cite{flajolet1985approximate}, who showed that $O(\log \log m + \log(1/\veps) + \log(1/\delta))$ bits of memory are sufficient to return a $(1 \pm \veps)$ approximation with success probability $1-\delta$. A recent result of~\cite{nelson2020optimal} shows that $O(\log \log m + \log(1/\veps) + \log\log(1/\delta))$ bits suffice for the same task. 

The original Morris counter is a random state machine with the following simple structure (desribed in the algorithm below): At each state $s = 1,2,3,\ldots,$ an increment causes the counter to transition to state $s+1$ with
probability $2^{-s}$, and to remain in state $s$ with probability $1-2^{-s}$. 
\begin{algorithm}
\caption{IncrementMorrisCounter}\label{alg:cap}
\begin{algorithmic}[1]
\Require  {Previous memory state} $s$
\Ensure {Next memory state} $s$
\State $B \gets \text{Draw from a }\mathsf{Bern}(2^{-s}) \text{ distribution}$
\State $s \gets s+ B$
\end{algorithmic}
\end{algorithm}\\
This is formally the discrete time pure birth process of Figure~\ref{fig:morris}:

\begin{center}
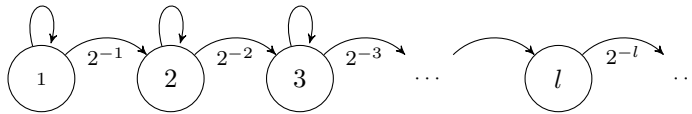
\begin{figure}[H]
\centering
\setlength\belowcaptionskip{-1.4\baselineskip}
\begin{tikzpicture}
  \tikzset{
    >=stealth',
    node distance=1.7cm,
    state/.style={font=\scriptsize,circle, align=center,draw,minimum size=25pt},
    dots/.style={state,draw=none}, edge/.style={->},
  }
  \node [state] (1) {$1$};
  \node [state ,label=center:$2$] (2)   [right of = 1] {};
  \node [state,label=center:$3$] (3) [right of = 2]   {};
  \node [dots]  (dots)  [right of = 3] {$\cdots$};
  \node [state,label=center:$l$] (l) [right of = dots]  {};
  \node [dots]  (dots2)  [right of = l] {$\cdots$};
  \path [->,draw,thin,font=\footnotesize]  (1) edge[loop above] node[below] {} (1);
  \path [->,draw,thin,font=\footnotesize]  (1) edge[bend left=45] node[below] {$2^{-1}$} (2);
  \path [->,draw,thin,font=\footnotesize]  (2) edge[loop above] node[below] {} (2);
  \path [->,draw,thin,font=\footnotesize]  (2) edge[bend left=45] node[below] {$2^{-2}$} (3);
   \path [->,draw,thin,font=\footnotesize]  (3) edge[loop above] node[below]  {} (3);
  \path [->,draw,thin,font=\footnotesize]  (3) edge[bend left=45] node[below] {$2^{-3}$} (dots);
   \path [->,draw,thin,font=\footnotesize]  (dots) edge[bend left=45] node[below] {} (l); 
  >,draw,thin,font=\footnotesize]  (l) edge[loop above] {} (l);
  \path [->,draw,thin,font=\footnotesize]  (l) edge[bend left=45] node[below] {$2^{-l}$} (dots2);
\end{tikzpicture}
\caption{The original Morris counter}\label{fig:morris}
\end{figure}
\end{center}
The performance of the above counter was characterized by Flajolet, who proved the following theorem.
\begin{theorem}[\cite{flajolet1985approximate}]~\label{thm:morris}
Let $C_m$ be the value of the Morris counter after $m$ increments. It holds that
\begin{align}
    \E (C_m)=\log m +\mu+g(\log m)+\phi(m),
\end{align}
 where 
 \begin{align}
    \mu = \frac{\gamma}{\ln 2}+\frac{1}{2}-\sum_{i=1}^{\infty}\frac{1}{2^i-1} \hspace{1mm} \text{and }\gamma=\lim _{n \rightarrow \infty}\left(-\log n+\sum_{k=1}^n \frac{1}{k}\right)\text{is Euler's constant},
\end{align}
$g(\cdot)$ is a periodic function of amplitude less than $10^{-5}$, $ |\phi(m)|\leq\min \left\{1,\frac{2^{\sqrt{16\log m}}\cdot (\log m)^{4.5}}{2m}\right\}$ and the expectation is over the randomness of the counter.\footnote{In~\cite{flajolet1985approximate}, Flajolet bounded $\phi(m)$ with $O(m^{-0.98})$. Here, we carefully follow the constants in his derivation and provide an explicit upper bound on the the error terms, since we are interested in bounds that can be applied for finite $m$.}
\end{theorem}
In his paper, Flajolet approximated $\E (C_m)$ with the Mellin integral transform of some function $\phi$ related to the marginal distribution of the counter, and then used Cauchy's residue theorem in order to compute the integral. The constant $\mu$ arises from the residual of the function at $0$, where it has a double pole. Thus, if we are interested in approximating $\log m$ using the counter, then using $C_m-\mu$ as our approximation guarantees that on average our additive error will not be more than $10^{-5}+\phi(m)$, a property that we leverage in our entropy estimation algorithm. 
\subsection{ Finite-State Bias Estimation Machine}
In the bias estimation problem, we are given access to i.i.d. samples drawn from the $\Bern(p)$ distribution, and we wish to estimate the value of $p$ under the expected quadratic loss (also known as mean squared error distortion measure).
The $S$-state randomized bias estimation algorithm presented below was purposed by~\cite{samaniego1973estimating}, and the Markov chain induced by algorithm is described in Figure~\ref{fig:samaniego}.
\begin{algorithm}
\caption{IncrementBiasEstimation}\label{alg:bias}
\begin{algorithmic}[1]
\Require {Number of states} $S$, {previous memory state} $s$, {a sample $X\sim\Bern(p)$}
\Ensure {Next memory state} $s$, {parameter estimate $\hat{p}$}
\State $B \gets \text{Draw from a }\mathsf{Bern}\left(\frac{s-1}{S-1}\right) \text{ distribution}$
\If{$X=1$}
\State $s \gets s+\overline{B}$
    \Else
    \State $s \gets s-B$    
\EndIf
\State $\hat{p}\gets \frac{s-1}{S-1}$
\end{algorithmic}
\end{algorithm} 
 \begin{center}
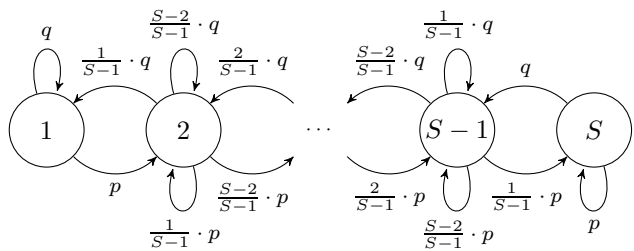
\begin{figure}[H]
\centering
\setlength\belowcaptionskip{-1.4\baselineskip}
\begin{tikzpicture}
  \tikzset{
    >=stealth',
    node distance=1.8cm,
    state/.style={font=\scriptsize,circle, align=center,draw,minimum size=28pt},
    dots/.style={state,draw=none}, edge/.style={->},
  }
  \node [state ,label=center:$1$] (1) {};
  \node [state ,label=center:$2$] (2)   [right of = 1] {};
  \node [dots]  (dots)  [right of = 2] {$\cdots$};
  \node [state ,label=center:$k$] (k)   [right of = dots] {};
  \node [dots]  (dots2)  [right of = k] {$\cdots$};
  \node [state,label=center:$S-1$] (S-1) [right of = dots2]  {};
  \node [state,label=center:$S$] (S) [right of = S-1]  {};
  \path [->,draw,thin,font=\footnotesize]  (1) edge[loop above] node[above] {$q$} (1);
  \path [->,draw,thin,font=\footnotesize]  (1) edge[bend right=45] node[below] {$p$} (2);
  \path [->,draw,thin,font=\footnotesize]  (2) edge[bend right=45] node[above] {$\frac{1}{S-1}\cdot q$} (1);
  \path [->,draw,thin,font=\footnotesize]  (2) edge[loop above] node[above] {$\frac{S-2}{S-1}\cdot q$} (2);
  \path [->,draw,thin,font=\footnotesize]  (2) edge[bend right=45] node[below] {$\frac{S-2}{S-1}\cdot p$} (dots);
  \path [->,draw,thin,font=\footnotesize]  (2) edge[loop below] node[below] {$\frac{1}{S-1}\cdot p$} (2);
  \path [->,draw,thin,font=\footnotesize]  (dots) edge[bend right=45] node[above] {$\frac{2}{S-1}\cdot q$} (2);
  \path [->,draw,thin,font=\footnotesize]  (dots) edge[bend right=45] node[below] {} (k);
  \path [->,draw,thin,font=\footnotesize]  (k) edge[bend right=45] node[above] {$\frac{k-1}{S-1}\cdot q$} (dots);
  \path [->,draw,thin,font=\footnotesize]  (k) edge[loop above] node[above] {$\frac{S-k}{S-1}\cdot q$} (k);
  \path [->,draw,thin,font=\footnotesize]  (k) edge[bend right=45] node[below] {$\frac{S-k}{S-1}\cdot p$} (dots2);
  \path [->,draw,thin,font=\footnotesize]  (k) edge[loop below] node[below] {$\frac{k-1}{S-1}\cdot p$} (k);
  \path [->,draw,thin,font=\footnotesize]  (dots2) edge[bend right=45] node[above] {} (k);
  \path [->,draw,thin,font=\footnotesize]  (dots2) edge[bend right=45] node[below] {$\frac{2}{S-1}\cdot p$} (S-1);
  \path [->,draw,thin,font=\footnotesize]  (S-1) edge[bend right=45] node[above] {$\frac{S-2}{S-1}\cdot q$} (dots2);
  \path [->,draw,thin,font=\footnotesize]  (S-1) edge[bend right=45] node[below] {$\frac{1}{S-1}\cdot p$} (S);
  \path [->,draw,thin,font=\footnotesize]  (S-1) edge[loop below] node[below] {$\frac{S-2}{S-1}\cdot p$} (S-1);
  \path [->,draw,thin,font=\footnotesize]  (S) edge[bend right=45] node[above] {$q$} (S-1);
  \path [->,draw,thin,font=\footnotesize]  (S-1) edge[loop above] node[above] {$\frac{1}{S-1}\cdot q$} (S-1);
  \path [->,draw,thin,font=\footnotesize]  (S) edge[loop below] node[below] {$p$} (S);
\end{tikzpicture}
\caption{Randomized bias estimation machine ($q=1-p$)}\label{fig:samaniego}
\end{figure}
\end{center}
The performance of this algorithm was carefully analyzed by~\cite{leighton1986estimating} where it was shown, using the Markov chain tree theorem, to asymptotically induce a $\mathrm{Binomial}(S-1,\theta)$ stationary distribution on the memory state space, which implies that $\E(\hat{p}-p)^2\leq O(1/S)$. In the same paper, it was further showed that the machine is order-optimal, by proving a lower bound of $\E(\hat{p}-p)^2\geq \Omega(1/S)$ for any finite-state estimator.
For completeness, we provide a simple proof for the MSE achieved by this construction in the appendix.
\begin{lemma}\label{lem:biasMSE}
     Let $\hat{p}(k)=\frac{k-1}{S-1}$ be the estimate of $p$ given state $k$ in the bias estimation machine of Figure~\ref{fig:samaniego}. Then we have
     \begin{align}
         \mathsf{MSE}=\lim_{t\rightarrow \infty}\E(\hat{p}(M_t)-p)^2\leq \frac{1}{S-1}.
     \end{align}
 \end{lemma}

\section{Upper Bound - Entropy Estimation Algorithm}\label{sec:upper}
In this section we prove Theorem~\ref{thm:upper_bound}, that is, we show the existence of an $S$-state randomized entropy estimation machine with $S \leq \frac{(c+1)^4n\cdot (\log n)^4}{\beta^2\delta}$ states that achieves additive error $\veps$ of at most $10^{-5}+\beta$. The basic idea is to let nature draw some $X$ from $p$ and use a Morris counter to approximate $-\log p_X$, then, since we are looking for $H(p)=\E(-\log p_X)$, use a bias estimation machine to simulate the averaging operation, by randomly generating coin tosses with bias that is proportionate to our estimate of $-\log p_X$. The bias estimation machine is incremented whenever a count is concluded in a randomized clock, which is simulated by another Morris counter. For a sufficiently large number of samples, this averaging converges (approximately) to the mean of $-\log p_X$, and thus outputs an approximation to the true underlying entropy. We divide our presentation to four parts: in the first part we describe the algorithm; in the second part we count the total number of states used by the algorithm; in the third part we assume the bias estimation machine is fed with an infinite number of i.i.d. samples and analyze the performance of the algorithm; and in the fourth part we relax this assumption by studying the mixing time of the Markovian process induced by our bias estimation machine. This allows us to prove an upper bound on the number of samples the developed algorithm requires.      
\begin{algorithm}
\caption{Entropy Estimation with Morris Counters}\label{alg:ent_est}
\begin{algorithmic}[1]
\Require  {A data stream} $X_1,X_2\ldots \sim p$,
{alphabet size} $n$, {run time} $t$, {error probability} $\delta$, $\beta>0$ ,$c>1$, {constant} $\mu$ 
\Ensure {Entropy estimate} $\hat{H}$\\
{Set} 
\begin{align}
& B\gets\min \{k\in \mathbb{N}:\lceil n^c \rceil\leq 2^k\}, \hspace{2mm} M \gets B+1\\&\eta \gets \text{Monte Carlo estimate of } \E (\log N) \text{ for } N=\sum_{k=1}^{M-1}\tau_k, \text{ where } \tau_k\sim \text{Geo}(2^{-k})\\&a\gets 1-\frac{\mu+\eta}{M},\hspace{2mm} S_{\text{bias}}\gets\left\lceil\frac{4M^2}{\beta^2\delta}\right\rceil+1\\& C_N\gets 1, \hspace{2mm} C_{N_x}\gets 1,\hspace{2mm} s\gets1 
\end{align}
\For {$i = 1,\ldots,t$}
\If{$C_{N}=1$} 
\State $x_{\text{test}}\gets X_i$
    \State $C_{N_x}\gets 1$
\Else
\State $C_N=\text{IncrementMorrisCounter}(C_N)$
\If {$X_i=x_{\text{test}}$}
\State $C_{N_x} \gets \text{IncrementMorrisCounter}(C_{N_x})$
\EndIf
\If{$C_{N_x}<2M$}  
\If {$C_N=M$} 
    \State $\theta_{N_x} \gets a-\frac{C_{N_x}-(\mu+\eta)}{2M}$
    \State $(s,\hat{\theta}) \gets$ \text{IncrementBiasEstimation}($S_{\text{bias}},s,\theta_{N_x}$) 
    \State $C_N\gets 1$
\EndIf
\Else 
\State $C_N\gets 1$
\EndIf
\EndIf
\EndFor
    \State $\hat{H}\gets 2M(\hat{\theta}-a)$
\end{algorithmic}
\end{algorithm}


\subsection{Description of the algorithm}
\begin{enumerate}
\item The algorithm receives an \textit{accuracy parameter} $\beta>0$ and an \textit{overhead parameter} $c>1$.
    \item In each iteration of the algorithm we collect a fresh sample $X\sim p$, and store its value. 
    Assuming the received sample is $x$, we proceed to estimate $\log p_x$ based on more fresh samples using Morris counters. 
    \item We use two Morris counters - one that approximates a clock, and one that approximates a count for $x$ values.
    \begin{itemize}
        \item The first counter is randomly incremented whenever a new sample is observed, and it stops when it reaches state $M=B +1$ states, where $B$ is the smallest integer $k$ such that $\lceil n^c \rceil\leq 2^k$. We denote the state of this counter as $C_N$. Let $N$ denote the \emph{random} time it takes  the counter to arrive at state $M$. We will show in the sequel that $N$ is expected to be around $\lceil n^c \rceil$ (up to small factors), thus this counter essentially approximates a clock that counts until $\lceil n^c \rceil$ samples are observed.
        \item The second counter is randomly incremented whenever $x$ is observed, and it stops when the first counter reaches state $M$.
        We denote the state of this counter as $C_{N_x}$.  We allow this counter to have $2M$ states to make sure the probability it ends before the first counter is sufficiently small. In the event that the counter indeed reaches state $2M$ \textit{before} the first counter, we draw a fresh sample and initialize both counters. This counter approximates the logarithm of the number of observed $x$ values in the length $N$ window.
    \end{itemize}
    \item Denoting the number of observed $x$ values in the previous stage as $N_x$, we define $\overline{C}_{N_x}\triangleq C_{N_x}-\mu-\E\log N$ to be the centralized output of the second counter. As $\E(\log N)$ is only a function of $n$, it can be calculated offline using Monte-Carlo simulation with the desired resolution (see Appendix). As we argue below, this is an almost unbiased estimator for $-\log p_x$. 
\item We now increment a bias estimation machine with $S_{\text{bias}}=\left\lceil\frac{4M^2}{\beta^2\delta}\right\rceil+1$ states whose purpose is to simulate the expectation operation.  
Specifically, each time the first Morris counter concludes a count, we generate a $\mathop{\mathrm{Ber}}(\theta_{N_x})$ random variable, with $\theta_{N_x}=a-\frac{ \overline{C}_{N_x}}{2M}$, and use it as the input to our bias estimation machine. The offset $a\triangleq 1-\frac{\E(\log N)+\mu}{2M}$ guarantees that $\theta_{N_x}\in [0,1)$ with probability $1$, as $\theta_{N_x}=1-\frac{C_{N_x}}{2M}$ and $1\leq C_{N_x}\leq 2M$ as it is the output of a Morris counter restricted to $2M$ states. Our estimator for the entropy $\hat{H}$ is the bias estimate of the machine, after subtraction of the known offset $a$ and multiplication by $2M$, that is, $\hat{H}=2M(\hat{\theta}-a)$.
\end{enumerate}
\subsection{Number of states in our machine}
As $n,t,\beta,\delta,M,\eta,a$ and $S_{\text{bias}}$ are program constants, we do not count them in the memory consumption of the algorithm. At each time point, our algorithm keeps the value of $x$, the state of the Morris counter approximating the clock, the state of the Morris counter approximating the logarithm of the $x$ counter, and the state of the bias estimation machine. Thus, the total number of states is the product of the individual number of states needed at each step, and recalling that $M=B+1\leq c\log n+2$, the total number of states is
\begin{align}
    S= n\cdot M \cdot 2M \cdot S_{\text{bias}}\leq 
    n\cdot 2M^2 \cdot \left(\frac{4M^2}{\beta^2\delta}+2\right) = n\left(\frac{8(c\log n +2)^4}{\beta^2\delta}+4(c\log n +2)^2\right).
\end{align}
\subsection{Analysis of the algorithm for $t=\infty$}
Let $X$ be the fresh sample collected at the start of an algorithm iteration. Our analysis is based on characterizing the joint distribution of $(N,N_X,C_{N_X})$, where $N$ is the number of observed samples until $C_N$ arrives at state $M$, $N_X$ is the number of times $X$ appeared in the $N$ sample window, and $C_{N_X}$ is the final state of the second counter. 
We first bound the probability that $N$ diverges significantly from $n^c$ in Lemma~\ref{lem:probN}, and then apply the result to upper bound $\E \left(N^{-\alpha}\right)$ in Lemma~\ref{lem:expN} for any $\alpha \in (0,1]$. We proceed to bound the difference between the expectation of $\overline{C}_{N_x}^\infty=C_{N_X}^\infty-\mu-\E\log N$ and $-H(p)$ in Lemma~\ref{lem:biasmall}, where $C_{N_X}^\infty$ denotes the state of an \textit{infinite} memory Morris counter. We show in Lemma~\ref{lem:Mor_bias} that the expected difference between $C_{N_X}$ and $C_{N_X}^\infty$ is $O\left((\log n)/n^c\right)$, thus proving that the absolute difference between $\E(\overline{C}_{N_x})$ and the entropy is at most $\psi_c(n)$. Lemma~\ref{lem:bias_mach} proves that the input to the bias estimation machine is an i.i.d. sequence of Bernoulli random variables, with parameter $\theta$ that equals to $\E(\overline{C}_{N_x})$, after scaling and shifting it to be in the interval $(0,1]$. This implies, by lemma~\ref{lem:bias_mach2}, that a judicious linear transformation of the output of the bias estimation machine is at most $\beta$ far from $\E(\overline{C}_{N_x})$ with probability at least $1-\delta$, thus giving the $(\veps,\delta)$ guarantee of Theorem~\ref{thm:upper_bound}.
\begin{lemma}\label{lem:probN}
For any $m\leq 2^{\ell}$ for some $1\leq \ell \leq M-1$, it holds that
\begin{align}
    \Pr(N<m)\leq e\cdot 2^{-\frac{1}{2}\cdot (M-\ell-1)^2}.
\end{align}
Furthermore, for any $m\geq \alpha\cdot 4n^c$, it holds that
\begin{align}
    \Pr(N>m)\leq 5e^{-\alpha}.
\end{align}
\end{lemma}
\begin{proof}
Let $\tau_k$ be the time it takes to move from state $k$ to state $k+1$ in the first Morris counter and note that $\tau_k\sim \text{Geo}(2^{-k})$. Thus, $N=\sum_{k=1}^{M-1} \tau_k$. The moment generating function of $\tau_k$ is
\begin{align}
    \mathbb{M}_{\tau_k}(s)\triangleq\E\left(e^{s\tau_k}\right)=\frac{1}{1+2^k(e^{-s}-1)},
\end{align}
and it is defined for all $s<-\ln (1-2^{-k})$.
The moment generating function of $N$ is therefore
\begin{align}
   \mathbb{M}_{N}(s)=\prod_{k=1}^{M-1} \mathbb{M}_{\tau_k}(s)=\prod_{k=1}^{M-1} \frac{1}{1+2^k(e^{-s}-1)},
\end{align}
and is defined for all $s<-\ln (1-2^{-(M-1)})=-\ln (1-2^{-B})$. 
We apply Chernoff bound to prove both results. For the first bound, we have $\Pr(N<m)\leq  e^{-sm}  \cdot \mathbb{M}_{N}(s) $ for any $s<0$. Setting $ s=-\ln (1+1/m)$, we get
\begin{align}
  \Pr(N<m)&\leq   \left(1+\frac{1}{m}\right)^m \cdot \prod_{k=1}^{M-1} \frac{1}{1+\frac{2^k}{m}}\leq e \cdot 2^{-\sum_{k=1}^{M-1}\log \left(1+\frac{2^k}{m}\right)}.
\end{align} 
As $m\leq 2^\ell$ for some $1\leq \ell \leq M-1$, we have $2^k/m\geq 2^{k-\ell}$, so we can lower bound the exponent above with
\begin{align}
   \sum_{k=1}^{M-1}\log (1+2^{k-\ell})&=\sum_{k=0}^{\ell-1}\log (1+2^{-k})+\sum_{k=1}^{M-\ell-1}\log (1+2^k)\\&\geq \sum_{k=1}^{M-\ell-1}k\\&\geq \frac{1}{2}\cdot (M-\ell -1)^2.
\end{align}
For the second bound, we have $ \Pr(N>m)\leq  e^{-sm}  \cdot \mathbb{M}_{N}(s)$ for any $s>0$. Setting $ s=-\ln (1-1/4n^c)$, for which $\mathbb{M}_{N}(t)$ is well defined as $2^B\leq 2n^c$, we get 
\begin{align}
  \Pr(N>m)&\leq   \left(1-\frac{1}{4n^c}\right)^m\cdot \prod_{j=1}^{M-1} \frac{1}{1-\frac{2^j}{4n^c}}\\&\leq \left(1-\frac{1}{4n^c}\right)^m\cdot  \prod_{j=1}^{M-1} \frac{1}{1-2^{-j}}\label{eq:jbound}\\& \leq 5\exp \left\{-\frac{m}{4n^c}\right\}=5e^{-\alpha}\label{eq:prod_bound}
\end{align} 
where in~\eqref{eq:jbound} we used the fact that $\frac{2^j}{4n^c}\leq \frac{2^j}{2T}=2^{j-M}$, and in~\eqref{eq:prod_bound} we used the
bound $\prod_{j=1}^{M}(1-2^{-j})\geq \frac{1}{4}+\frac{1}{2^{M+1}}$, which can be proved via induction.

\end{proof}
\begin{lemma}\label{lem:expN}
Let $v_n(\alpha)\triangleq\sqrt{\frac{2c\alpha^3}{\log n}}+\frac{\alpha}{\log n}$. Then for any $0< \alpha \leq 1$, we have that 
 \begin{align}
    \E \left(N^{-\alpha}\right) \leq (e+1)n^{-c\cdot \alpha +v_n(\alpha)}.
\end{align}
\end{lemma}
\begin{proof}
    Appealing to Lemma~\ref{lem:probN}, for any $1\leq \ell \leq M-1,m\leq 2^\ell$, we have
  \begin{align}
    \E(N^{-\alpha}) &\leq \Pr (N<m)+m^{-\alpha}\cdot \Pr (N\geq m) \\&\leq e\cdot 2^{-\frac{1}{2}\cdot (M-\ell-1)^2}+2^{-\ell\cdot\alpha}.
\end{align}  
Setting $\ell = M-1-\left\lceil\sqrt{2\alpha \cdot (M-1)}\right\rceil$ and recalling that $n^c\leq 2^{M-1}$, we get
\begin{align}
   \E(N^{-\alpha}) &\leq e\cdot 2^{-\frac{1}{2}\left\lceil\sqrt{2\alpha(M-1)}\right\rceil^2}+2^{-\alpha (M-1)\cdot \left(1-\frac{\left\lceil\sqrt{2\alpha(M-1)}\right\rceil}{M-1}\right)}\\&\leq e\cdot
   2^{-\alpha (M-1)}+2^{-\alpha (M-1)\cdot \left(1-\frac{\sqrt{2\alpha(M-1)}+1}{M-1}\right)}
 \leq e\cdot n^{-c\cdot\alpha}+n^{-c\cdot \alpha+v_n(\alpha)}.
\end{align}
\end{proof}
According to Theorem~\ref{thm:morris}, the value of the infinite memory Morris counter after $m$ updates is close to $\log m$ in expectation, up to some small bias. We would like to show that, given $N$ and $N_X$, the expectation of our counter is close to the expectation of $\log(N_X)$, which is the expectation of the logarithm of a Binomial variable, and that centering and taking the expectation over $(N,X)$ gives us approximately $-H(p)$. To that end, we first analyze the algorithm under the assumption that the second counter is an infinite memory Morris counter, which we denote as $C_{N_X}^\infty$, and then we prove that the expected gap between $C_{N_X}^\infty$ and $C_{N_X}$ is small.
\begin{lemma}\label{lem:biasmall}
 Let $\phi_c(n)=(e+1)n^{-(c-1)+v_n(1)}+\min\{1, C\cdot  n^{-\frac{1}{2}\cdot (c-1)+v_n(1/2)}\}$, where $C=2(e+1)\cdot 10^{8}$. Then
\begin{align}
  |\E (\overline{C}_{N_x}^\infty) +H(p)|\leq 10^{-5}+\phi_c(n).  
\end{align}   
\end{lemma} 
\begin{proof}
As $C_{N_X}^\infty$ is the state of the infinite memory Morris counter, we have that $\overline{C}_{N_x}^\infty= C_{N_x}^\infty-\mu-\E\log N$. Given $N$ and $X=x$, the number of $x$ observations in the $N$ sample window is distributed  Binomial($N,p_x$). 
We show that $\E(\overline{C}_{N_x}^\infty)$ is close to the expected logarithm of the  normalized Binomial random variable, which then gives us $-H(p)$ plus some bias. From Theorem~\ref{thm:morris}, we have
\begin{align}
    \E(\overline{C}_{N_x}^\infty\mid X,N_X)=\E(C_{N_X}^\infty-\mu -\E(\log N)\mid X,N_X)= \log N_X-\E(\log N)+g(\log N_X) +\phi(N_X).
\end{align}
Note that
\begin{align}
  \E(\log N_X\mid X=x)-\E(\log N)&=\sum_{k=1}^\infty \Pr(N_x=k)\log k - \sum_{m=1}^\infty \Pr(N=m)\log m  \\& =\sum_{m=1}^\infty\sum_{k=1}^m\Pr(N=m,N_x=k) \log k-\sum_{m=1}^\infty \Pr(N=m)\log m \\&=\sum_{m=1}^\infty\Pr(N=m)\sum_{k=1}^m \Pr(N_x=k\mid N=m)\log \frac{k}{m} \\&=\E\left(\frac{N_X}{N}\mid X=x\right).
\end{align}
Thus, letting $\gamma_{N_x}=g(\log N_x) +\phi(N_x)$, we have
\begin{align}
  \E(\overline{C}_{N_x}^\infty\mid X)=\E\left(\log \frac{N_X}{N}\mid X\right)+\E(\gamma_{N_X}\mid X) =\log p_X+\E\left(\log \frac{N_X}{N\cdot p_X}\mid X\right)+\E(\gamma_{N_X}\mid X),
\end{align}
implying that $    \E(\overline{C}_{N_x}^\infty)=-H(p)+\E\left(\log \frac{N_X}{N\cdot p_X}\right)+\E(\gamma_{N_X})$.
We conclude the proof by bounding $\E\left(\log \frac{N_X}{N\cdot p_X}\right)$ in  Lemma~\ref{lem:smallratio}, and then bounding $\E(\gamma_{N_X})$ in Lemma~\ref{lem:gammaB}.
\end{proof}

\begin{lemma}\label{lem:smallratio}
It holds that
\begin{align}
   0\leq \E\left(\log \frac{N_X}{N\cdot p_X}\right)\leq (e+1)n^{-(c-1)+v_n(1)}.
\end{align}
\end{lemma}
\begin{proof}
  We first show that $\E\left(\log \frac{N_X}{N\cdot p_X}\right)\geq 0$. By Jensen's inequality and convexity of $t\mapsto-\log(t)$
  \begin{align}
   \E\left(\log \frac{N_X}{N\cdot p_X}\right)&=\E_{X,N}\left[\E_{N_X|N,X}\left(-\log \frac{N\cdot p_X}{N_X}\right)\right] \\&\geq -\E_{X,N}\left[\log \left(N\cdot p_X\cdot\E_{N_X|N,X}\left[\frac{1}{N_X}\right]\right)\right].
\end{align} 
To establish non-negativity of $\E\left(\log \frac{N_X}{N\cdot p_X}\right)$, it therefore suffices to show that $\E_{N_X|N,X=x}\left[\frac{1}{N_X}\right]\leq \frac{1}{p_x\cdot N}$. To that end, recall that given $X=x$ and $N$, we have $N_X \sim \mathop{\mathrm{Bin}}(N-1,p_x)+1$. Thus, we indeed have
    \begin{align}
        \E_{N_X|N,X=x}\left[\frac{1}{N_X}\right]&=\sum_{m=0}^{N-1}\frac{1}{m+1}{N-1 \choose m} p_x^m (1-p_x)^{N-m-1}\nonumber\\&=\sum_{m=0}^{N-1}\frac{1}{p_x\cdot N}{N \choose m+1} p_x^{m+1} (1-p_x)^{N-m-1}\nonumber\\&=\frac{1-(1-p_x)^{N}}{p_x\cdot N}\nonumber\\&\leq \frac{1}{p_x\cdot N}.\label{eq:bin_bound}
    \end{align}
 To upper bound $\E\left(\log \frac{N_X}{N\cdot p_X}\right)$, we use Jensen's inequality and the concavity of $t\mapsto\log t$, to obtain
\begin{align}
  \E_{N_X|N,X=x}\left(\log \frac{N_X}{N\cdot p_X}\right)&\leq  \log \left(\frac{\E_{N_X|N,X=x}[N_X]}{N\cdot p_x}\right) \\&= \log \left(1+\frac{1-p_x}{N\cdot p_x}\right)\\&\leq \frac{1}{N\cdot p_x}.
\end{align}
Thus, overall, 
\begin{align*}
\E\left[\log \frac{N_X}{N\cdot p_X}\right]&\leq\E_{N,X} \left[\frac{1}{N\cdot p_X}\right]\\&= \E_{X} \left[\frac{1}{p_X}\right]\E_{N} \left[\frac{1}{N}\right]\\&=n\cdot \E_{N} \left[\frac{1}{N}\right]
\end{align*}
and appealing to Lemma~\ref{lem:expN} with $\alpha=1$, we have $\E\left[\log \frac{N_X}{N\cdot p_X}\right]\leq (e+1)n^{-(c-1)+v_n(1)}$. 
\end{proof}

\begin{lemma}\label{lem:gammaB}
It holds that
\begin{align}
   \E (\gamma_{N_X})\leq 10^{-5}+\min\{1, C\cdot  n^{-\frac{1}{2}\cdot (c-1)+v_n(1/2)}\}.
\end{align}
\end{lemma}
\begin{proof}
Note that $\E (g(\log N_x))\leq 10^{-5}$ is explicit in Theorem~\ref{thm:morris} for any $x\in[n]$, and in particular, $\E (g(\log N_X))\leq 10^{-5}$. Thus, it remains to upper bound $\E(\phi(N_X))$. It is straightforward to verify that $\phi(x)\leq \min \left\{1,\frac{2\cdot 10^{8}}{\sqrt{x}}\right\}$ for all $x\geq 1$, and consequently,
\begin{align}
\E(\phi(N_X))\leq \E\left[\min \left\{1,\frac{2\cdot 10^{8}}{\sqrt{N_X}}\right\}\right]\leq \min \left\{1,2\cdot 10^{8}\E\left[\sqrt{\frac{1}{N_X}}\right]\right\}.
\label{eq:EPhibound}
\end{align}
From Jensen's inequality, concavity of $t\mapsto\sqrt{t}$, and equation~\eqref{eq:bin_bound},
    \begin{align}
       \E\left[\sqrt{\frac{1}{N_X}}\right]&=\E_{N,X}\left[\E_{N_X|N,X}\left[\sqrt{\frac{1}{N_X}}\right]\right]\\&\leq \E_{N,X}\left[\sqrt{\E_{N_X|N,X}\left[\frac{1}{N_X}\right]}\right]\\&\leq \E_{N,X}\left[\sqrt{\frac{1}{p_X\cdot N}}\right]\nonumber\\
       &=\E_{N}\left[\sqrt{\frac{1}{ N}}\right]\E_{X}\left[\sqrt{\frac{1}{p_X}}\right].
       \label{eq:ENxinv}
    \end{align}
Note that, again using Jensen's inequality and concavity of $t\mapsto\sqrt{t}$, we have
\begin{align}
\E_{X}\left[\sqrt{\frac{1}{p_X}}\right]=\sum_{x=1}^n\sqrt{p_x}\leq n\sqrt{\frac{1}{n}\sum_{x=1}^n p_x} =\sqrt{n}.  
\label{ex:EPxinv}
\end{align}
Appealing to Lemma~\ref{lem:expN} with $\alpha=0.5$, we have
\begin{align}
\E(N^{-0.5})\leq (e+1)n^{-\frac{c}{2}+v_n(1/2)}.    
\label{eq:ENisqrt}
\end{align}
Thus, substituting~\eqref{ex:EPxinv} and~\eqref{eq:ENisqrt} into~\eqref{eq:ENxinv} and then into~\eqref{eq:EPhibound}, and recalling that $C=2(e+1)10^8$, we obtain the claimed result.
\end{proof}
Lemma~\ref{lem:Mor_bias} below bounds the absolute difference between the expectation of the Morris counter and the expectation of the truncated Morris counter of the algorithm by $O((\log n) /n^c)$. The proof is relegated to the appendix.  \begin{lemma}\label{lem:Mor_bias}
    We have
    \begin{align}
 |\E(C_{N_X}^\infty) -\E(C_{N_X})|\leq n^{-c}\cdot\frac{ 100(c\log n+2)}{(1-0.5n^{-c})^2} 
\end{align}
 \end{lemma}  
We now turn to analyzing the bias estimation phase of the algorithm.
\begin{lemma}\label{lem:bias_mach}
   Let $Y_{N_1},Y_{N_2},\ldots$ denote the sequence of Bernoulli random variables fed to the bias estimation machine. Then
   \begin{align}
       Y_{N_1},Y_{N_2},\ldots\overset{i.i.d.}{\sim}\mathsf{Bern}\left(\theta\right),
   \end{align}
   where $\theta=\frac{H(p)+b}{2M}+a$ and $|b|\leq 10^{-5}+\psi_c(n)$.
\end{lemma}
\begin{proof}
    The sequence of samples is i.i.d. since each sample is a function of the i.i.d. series $\{X_i\}_{i=1}^{\infty}$ and the statistics of the Morris counters, which are initialized at every incrementation. Given $(X,N,N_X,C_{N_X})=(x,m,n_x,C_{N_x})$, we set the Bernoulli parameter $\theta_{N_x}=a-\frac{ \overline{C}_{N_x}}{2M}$. Thus the unconditioned parameter $\theta$ is a mixture of $\theta_{N_X}$ over the joint distribution of $(X,N,N_X,C_{N_X})$, that is,
    \begin{align}
     \theta=\E(\theta_{N_X})=a-\frac{\E(\overline{C}_{N_x})}{2M}
     =a+\frac{H(p)+b}{2M},
    \end{align}
    where we used Lemma~\ref{lem:biasmall} and Lemma~\ref{lem:Mor_bias}.
\end{proof}
\begin{lemma}\label{lem:bias_mach2}
   We have
  \begin{align}
    \Pr(|\hat{H}-(H(p)+b)|>\beta)\leq\delta.
  \end{align}
\end{lemma}
\begin{proof}
Recall that $\hat{H}=2M(\hat{\theta}-a)$, and that $\E(\hat{\theta}-\theta)^2\leq \frac{1}{S_{\text{bias}}-1}$. As $S_{\text{bias}}=\left\lceil\frac{4M^2}{\beta^2\delta}\right\rceil+1$, we have 
\begin{align}
    \E(\hat{H}-(H(p)+b))^2=4M^2\cdot \E(\hat{\theta}-\theta)^2\leq \beta^2\delta,
\end{align}
thus, from Chebyshev's inequality,
\begin{align}
    \Pr(|\hat{H}-(H(p)+b)|>\beta)\leq \frac{\E(\hat{H}-(H(p)+b))^2}{\beta^2}\leq\delta.
\end{align}
\end{proof}
Note that our upper bound on the additive error in estimation of $H(p)$ is $\beta+|b|\leq  \beta+10^{-5}+\psi_c(n)$, which limits our results to estimation error $\veps>10^{-5}+\psi_c(n)$.  

\subsection{Analysis of the algorithm for $t<\infty$}
In the previous analysis, the number of observed samples was assumed to be unbounded. In practice we only need to observe $O(t_{\mathsf{mix} }(\theta))$ samples, where $t_{\mathsf{mix} }(\theta)$ is the mixing time of our machine whenever the input is $\mathop{\mathrm{Bern}}(\theta)$ samples, i.e., the minimal time it takes for the total variation distance between the marginal distribution and the limiting distribution to be small. Lemma~\ref{lem:mix} and Lemma~\ref{lem:total_run} bound the number of samples needed at the Morris counting phase and characterize the mixing time of the bias estimation machine, respectively. Combining the previous results, Lemma~\ref{lem:sample_comp} shows that the total run time of the algorithm needed to obtain an $\veps$ additive approximation of the entropy with probability at least $1-3\delta$ is as prescribed by Theorem~\ref{thm:upper_bound}. 

Specifically, recall that the bias estimation machine is only incremented after an iteration of the first Morris counter is completed, and the run time of each iteration is a random variable that is only bounded in expectation. We note that this in fact implies the existence of a good algorithm that has a bounded sample complexity; namely, running our entropy estimation algorithm on $L$ samples is equivalent to running the bias estimation machine from~\cite{samaniego1973estimating} on a random number of samples $k=k(L)$ times with $\theta=\E(\theta_{N_X})$. The randomness in $k(L)$ follows since the runtime $N_i$ of each iteration of the Morris counter procedure is a random variable. We use Chernoff's bound to upper bound the probability that $k(L)$ is small. This event is considered as an error in our analysis. We now upper bound the mixing time of the bias estimation machine from~\cite{samaniego1973estimating}. Whenever $k(L)$ is greater than this mixing time, the error of our algorithms with $L$ samples is close to its asymptotic value.

To upper bound the mixing time, we use the \textit{coupling method}. Recall that the transition matrix $P$ of a Markov process $\{X_t\}_{t=0}^{\infty}$ supported on $\mathcal{X}$ is a matrix whose elements are $\Pr(X_{t+1}=x'| X_{t}=x)=P(x,x')$, for any  $x,x'\in\mathcal{X}\times \mathcal{X}$. 
We define a coupling of Markov chains with transition matrix $P$ to be a
process $\{X_t,Y_t\}_{t=0}^{\infty}$ with the property that both $\{X_t\}_{t=0}^{\infty}$ and $\{Y_t\}_{t=0}^{\infty}$ are Markov chains
with transition matrix $P$, although the two chains may be correlated and have different
initial distributions. Given a Markov chain on $\mathcal{X}$ with transition matrix $P$, a \textit{Markovian coupling}
of two $P$-chains is a Markov chain $\{X_t,Y_t\}_{t=0}^{\infty}$ with state space $\mathcal{X}\times \mathcal{X}$, which
satisfies, for all $x,y,x',y'$,
\begin{align}
  \Pr(X_{t+1}=x'| X_{t}=x,Y_{t}=y)&=P(x,x')\\\Pr(Y_{t+1}=y'| X_{t}=x,Y_{t}=y)&=P(y,y').
\end{align}
Let $P^t(x_0)$ be the marginal distribution of the chain at time $t$ when initiated at $x_0$, and let $\pi$ be the unique stationary distribution. Define the $\delta$-mixing time as
\begin{align}
    t_{\delta}^* \triangleq \min \{t: \dtv(P^t(x_0),\pi)\leq \delta\},
\end{align}
and $t_{\mathsf{mix} }\triangleq t_{1/4}^*$. We now show that the bias estimation machine with $S$ states mixes in $\Theta(S\log S)$ time, uniformly for all $\theta\in (0,1]$.
\begin{lemma}\label{lem:mix}
   Let $t_{\mathsf{mix} }(p)$ denote the mixing time of the bias estimation machine with $S$ states when the input is i.i.d. $\Bern(p)$, and
   define the \textit{worst-case} mixing time to be $t^*=\max_{p\in (0,1]}t_{\mathsf{mix} }(p)$. Then 
   \begin{align}
      \ln (2)\cdot (S-1)\log (S-1)\leq t^*\leq  4S\log S. 
   \end{align}
\end{lemma}
\begin{proof}
The transition probabilities of the bias estimation machine of Figure~\ref{fig:samaniego} are given, for $1<k<S$, as
\begin{align}
    X_{t+1}|_{X_{t}=k}=\begin{cases}
        k+1,& \textit{w.p. } \frac{S-k}{S-1}\cdot p, \\k,& \textit{w.p. } \frac{k-1}{S-1}\cdot p+\frac{S-k}{S-1}\cdot q, \\k-1,& \textit{w.p. } \frac{k-1}{S-1}\cdot q,\end{cases}
\end{align}
and for the extreme states $\{1,S\}$ as
\begin{align}
    X_{t+1}|_{X_{t}=1}=\begin{cases}
        2,& \textit{w.p. } p, \\1,& \textit{w.p. } q,\end{cases}
   \hspace{4mm} X_{t+1}|_{X_{t}=S}=\begin{cases}
        S,& \textit{w.p. } p, \\S-1,& \textit{w.p. } q.\end{cases}    
\end{align}
We construct a Markovian coupling in which the two chains stay together at all times after their first simultaneous
visit to a single state, that is
\begin{align}
    \text{if } X_s=Y_s \text{ then } X_t=Y_t \text{ for all } t\geq s. \label{eq:couple}
\end{align}
The following theorem is due to~\cite{levin2017markov}(Theorem 5.4), will give us an upper bound on the mixing time using this coupling. 
\begin{theorem}
 Let $\{(X_t,Y_t)\}$ be a Markovian coupling satisfying~\eqref{eq:couple}, for which $X_0=x_0$ and $Y_0=y_0$. Let $\tau_{\mathsf{couple}}$ be the coalescence time of the chains, that is,
 \begin{align}
    \tau_{\mathsf{couple}}\triangleq \min \{t:X_t=Y_t\}. 
 \end{align}
 Then
\begin{align}
    t_{\mathsf{mix} }\leq 4 \max_{x_0,y_0\in \mathcal{X}}\E(\tau_{\mathsf{couple}}).
\end{align}    
\end{theorem}\label{thm:couple}
Assume w.l.o.g. that $x_0<y_0$ and let $U_t$ be an i.i.d. sequence drawn according to the $ \mathop{\mathrm{Unif}}(0,1)$ distribution.  We construct a coupling on $(X_t,Y_t)$ such that, at each time point $t<\tau_{\mathsf{couple}}$, $X_t$ and $Y_t$ are incremented in the following manner:
\begin{align}
    X_{t+1}|_{X_{t}=i}=\begin{cases}
        i+1,& \textit{if  }U_t\leq \frac{S-i}{S-1}\cdot p, \\i,& \textit{if  } \frac{S-i}{S-1}\cdot p\leq U_t\leq 1-\frac{i-1}{S-1}\cdot q, \\i-1,& \textit{if  } 1-\frac{i-1}{S-1}\cdot q\leq U_t \leq 1,\end{cases}
\end{align}
and
\begin{align}
    Y_{t+1}|_{Y_{t}=j}=\begin{cases}
        j+1,& \textit{if  }U_t\leq \frac{S-j}{S-1}\cdot p, \\j,& \textit{if  } \frac{S-j}{S-1}\cdot p\leq U_t\leq 1-\frac{j-1}{S-1}\cdot q, \\j-1,& \textit{if  } 1-\frac{j-1}{S-1}\cdot q\leq U_t \leq 1.\end{cases}
\end{align}
One can validate that the transition probabilities are the correct ones, for example
\begin{align}
 \Pr(X_{t+1}=i|X_{t}=i)&=\Pr\left( \frac{S-i}{S-1}\cdot p\leq U_t\leq 1-\frac{i-1}{S-1}\cdot q\right)\\&=1-\frac{i-1}{S-1}\cdot q- \frac{S-i}{S-1}\cdot p\\&=\frac{i-1}{S-1}\cdot p+\frac{S-i}{S-1}\cdot q,
\end{align}
and, similarly, $\Pr(Y_{t+1}=j|Y_{t}=j)=\frac{j-1}{S-1}\cdot p+\frac{S-j}{S-1}\cdot q$. The other transition probabilities are easily calculated. Note that $i<j$ implies $\frac{S-j}{S-1}<\frac{S-i}{S-1}$, thus $Y_t$ cannot move right unless $X_t$ moves right and $X_t$ cannot move left unless $Y_t$ moves left. Moreover, since $x_0<y_0$, we have $i<j$ for all $t<\tau_{\mathsf{couple}}$. This follows from construction, since $\frac{S-i}{S-1}\cdot p$ is always smaller than $1-\frac{j-1}{S-1}\cdot q$, implying that $X_t$ cannot jump over $Y_t$ when they are one-state apart. 
Thus, the \textit{distance} process $D_t\triangleq Y_t-X_t$, is a non-increasing function of $t$, with initial state $D_0=y_0-x_0$, that can only decrease by one unit at a time or stay unchanged. We have
\begin{align}
    \Pr(D_{t+1}=D_t-1)&=\Pr(X_{t+1}=X_t+1,Y_{t+1}=Y_t)+\Pr(Y_{t+1}=Y_t-1,X_{t+1}=X_t)\\&=\Pr\left(\frac{S-Y_t}{S-1}\cdot p \leq U_t \leq \frac{S-X_t}{S-1}\cdot p\right)+ \Pr\left(1-\frac{Y_t-1}{S-1}\cdot q \leq U_t \leq 1-\frac{X_t-1}{S-1}\cdot q\right)\\&=\frac{Y_t-X_t}{S-1}\cdot p+\frac{Y_t-X_t}{S-1}\cdot q\\&=\frac{D_t}{S-1}.
\end{align}
The expected coupling time is now the expected time it takes for $D_t$ to decrease from $D_0$ to $D_t$, thus in order to maximize it under the given coupling, we need to maximize $D_0$, which corresponds to setting $X_0=1,y_0=S$. For $D_0=S-1$, consider the process $M_t\triangleq D_0-D_t$, which is a non-decreasing function of $t$ that goes from $0$ to $S-1$ and has $\Pr(M_{t+1}=M_t+1)=\Pr(D_{t+1}=D_t-1)=\frac{D_t}{S-1}=1-\frac{M_t}{S-1}$. Then this process is no other than the \textit{Coupon Collector} process with $S-1$ coupons, and the expected coupling time in our chain in identical to the expected number of coupons collected until the set contains all $S-1$ types, which according to~\cite{levin2017markov}, Proposition 2.3., is
\begin{align}
    \E(\tau_{\mathsf{couple}})=(S-1)\cdot\sum_{k=1}^{S-1}\frac{1}{k}\leq (S-1)(\ln (S-1)+1)\leq S \log (S).\label{eq:mix_upper}
\end{align} 
To show that this upper bound is indeed tight, consider the case of $p=1$. In this case, the chain of Figure~\ref{fig:samaniego} is simply the Coupon Collector process with $S-1$ coupons, thus, letting $\tau$ be the (random) time it takes to collect all coupons, we have
\begin{align}
   \E(\tau)=(S-1)\cdot\sum_{k=1}^{S-1}\frac{1}{k}\geq \ln (2)\cdot (S-1)\log (S-1). 
\end{align} 
\end{proof}
From~\cite{levin2017markov}, Eq. (4.34), we have that the $\delta$-mixing time $t_{\delta}^*$ can be upper bounded in terms on the mixing time by
\begin{align}
   t_{\delta}^*\leq \left\lceil \log \left(\frac{1}{\delta}\right) \right\rceil \cdot t_{\mathsf{mix} }. \label{eq:del_mix_upper}
\end{align}
Let 
\begin{align}
    k\triangleq 4\left\lceil\log \left(\frac{1}{\delta}\right)\right\rceil\left(\frac{4(c\log n+2)^2}{\beta^2\delta}+1\right)\log \left(\frac{4(c\log n+2)^2}{\beta^2\delta}+1\right), 
\end{align}
and note that from equation~\eqref{eq:del_mix_upper}, Lemma~\ref{lem:mix}, and substituting $S_{\text{bias}}=\frac{4M^2}{\beta^2\delta}+1$, we have that the $\delta$-mixing time of the bias estimation machine is at most $ k$. Let $N_1,N_2,\ldots,N_k$ be the first $k$ i.i.d. Morris counter running times, which are all distributed as $N$ in the analysis from Section~\ref{sec:upper}. Lemma~\ref{lem:total_run} uses the concentration of $N$ to show that, with probability $1-\delta$, the number of samples we need to observe until the bias machine mixes is not large.
\begin{lemma}\label{lem:total_run}
  Let $m=4n^c \cdot \ln \left(\frac{5k}{\delta}\right)$. Then
   \begin{align}
       \Pr\left(\sum_{i=1}^kN_i>k\cdot m\right) \leq \delta.
   \end{align}
\end{lemma}
\begin{proof}
Appealing to Lemma~\ref{lem:probN} we have $   \Pr(N>m)\leq \delta/k$. Consequently, the probability that at least one of the random variables $N_1,\ldots, N_k$ is greater than $m$ is at most $1-\left(1-\frac{\delta}{k}\right)^k\leq \delta$.
\end{proof}
We conclude with the following lemma, which connects Lemma~\ref{lem:mix} and Lemma~\ref{lem:total_run} to show that our entropy estimator performs well even if the number of input samples is limited to $\tilde{O}(n^c/\delta)$.
\begin{lemma}\label{lem:sample_comp}
  Let the algorithm of Theorem~\ref{thm:upper_bound} run on $L=k\cdot m$ samples, and output the estimate $\hat{H}_{M_L}$. Then with probability at least $1-3\delta$, $\hat{H}_{M_L}$ is within $\veps$-additive error from $H(p)$.   
\end{lemma}
\begin{proof}
 Lemma~\ref{lem:total_run} implies that, with probability at least $1-\delta$, after observing $k\cdot m$ samples, the bias estimation machine has been incremented at least $k$ times. Recall that, by definition, after $t\geq t_{\delta}^*$ increments of the bias estimation machine, we have that $\dtv(P^t(x_0),\pi)\leq \delta$, and that our $S$-states entropy estimator has $\sum_{i\in \hat{H}_{\bar{\veps}}}\pi_i<\delta$, where $\hat{H}_{\bar{\veps}}=\{i\in [S]:|\hat{H}_i-H(p)|>\veps\}$. Thus, from a union bound, a fraction of $2\delta$ of the distribution $P^t(x_0)$ (at most) is supported on $\hat{H}_{\bar{\veps}}$. Putting it all together, we have that a finite-time algorithm that outputs an estimate $\hat{H}(M_L)$ after 
\begin{align}
  L=k\cdot m = \Omega \left(\frac{n^c\cdot \mathop{\mathrm{poly}}(\log n)}{\delta}\cdot \mathop{\mathrm{poly}}(\log (1/\delta))\right)  
\end{align}
will be $\veps$-far from the correct entropy with probability at most $3\delta$.~\footnote{Note that the probability that the second Morris counter saturates even once in $k$ iterations is less than $O \left(n^{-c}\cdot\mathop{\mathrm{poly}}(\log n)\frac{\mathop{\mathrm{poly}}(\log (1/\delta))}{\delta}\right)$, thus is negligible for any $\delta\gg n^{-c}$.} 
\end{proof}

\section{Lower Bounds}
In this section we prove Theorem~\ref{thm:lower_bound}. The $\Omega(n)$ bound is proved via reduction to uniformity testing. For the $\frac{\log n}{2\veps}$ bound, we use a simple quantization argument. Assume that $S<\frac{\log n}{2\veps}$. Then there must be two consecutive estimate values $\hat{H}_1,\hat{H}_2\in [0,\log n]$ such that $\hat{H}_2-\hat{H}_1> 2\veps$. This implies that $H=(\hat{H}_1+\hat{H}_2)/2$ has $|H-\hat{H}_1|=|H-\hat{H}_2|> \veps$. Thus, for this value of the entropy, we have $\Pr (|\hat{H}(M_t)-H|>\veps)=1$ for all $t\in \mathbb{N}$.

\subsection{Proof of the \texorpdfstring{$(1-2\sqrt{\veps\ln 2})n$}{Lg} bound}
An ($\veps, \delta$) uniformity tester can distinguish (with probability $0<\delta<1/2$) between the case where $p$ is uniform and the case where $p$ is $\veps$-far from uniform in total variation. Assume we have an ($\veps, \delta$) entropy estimator. Then we can obtain an ($\tilde{\veps}=\sqrt{\veps\ln 2}, \delta$) uniformity tester using the following protocol: the tester declares that $p$ is uniform if $\hat{H}>\log n - \veps$, and that $p$ is $\tilde{\veps}$-far from uniform if $\hat{H}<\log n - \veps$. We now argue that this is indeed an ($\tilde{\veps}, \delta$) uniformity tester, in which case the $(1-2\tilde{\veps}) n$ lower bound will follow immediately from the lower bound on uniformity testing of~\cite{berg2022memory}. If $p=u$, where $u$ is the uniform distribution over $[n]$, then $H(p)=\log n$ and $\hat{H}>\log n -\veps$ with probability at least $1-\delta$, so our tester will correctly declare ``uniform'' with probability at least $1-\delta$. If $\dtv(p,u)>\sqrt{\veps\ln 2}$, then from Pinsker's inequality (~\cite{cover2012elements}, Lemma $11.6.1$),
\begin{align}
    2\veps <\frac{2}{\ln 2} \dtv(p,u)^2\leq D(p||u)=\log n - H(p),
\end{align}
which implies $H(p)< \log n -2 \veps$ and $\hat{H}<\log n -\veps$ with probability at least $1-\delta$. Thus, our tester will correctly declare ``far from uniform'' with probability at least $1-\delta$.

\section{Memory Complexity of Mutual Information Estimation}
We extend our results to the problem of mutual information estimation. The upper bound follows by a slight tweaking of our entropy estimation machine, and the lower bound follows by noting the close relation between mutual information and joint entropy, and lower bounding the memory complexity of the latter. 
\subsection{Upper Bound achieving algorithm}
\begin{algorithm}
\caption{Mutual Information Estimation with Morris Counters}\label{alg:mut_est}
\begin{algorithmic}[1]
\Require  {A data stream} $(X_1,Y_1),(X_2,Y_2)\ldots \sim p_{XY}$,
{alphabet size} $n$, {alphabet size} $m$, {run time} $t$, {error probability} $\delta$, $\beta>0$ ,$c>1$, {constant $\mu$} 
\Ensure {Mutual Information estimate} $\hat{I}$\\
{Set} 
\begin{align}
& B\gets\min \{k\in \mathbb{N}:\lceil (nm)^c \rceil\leq 2^k\}, \hspace{2mm} M \gets B+1\\&\eta \gets \text{Monte Carlo estimate of } \E (\log N) \text{ for } N=\sum_{k=1}^{M-1}\tau_k, \text{ where } \tau_k\sim \text{Geo}(2^{-k})\\&a\gets \frac{2}{3}-\frac{\mu+\eta}{6M},\hspace{2mm} S_{\text{bias}}\gets\left\lceil\frac{36M^2}{\beta^2\delta}\right\rceil+1\\&C_N\gets 1, \hspace{2mm} C_{N_x},C_{N_y},C_{N_{xy}}\gets 1,\hspace{2mm} s\gets1 
\end{align}
\For {$i = 1,\ldots,t$}
\If{$C_{N}=1$} 
    \State $(x_{\text{test}},y_{\text{test}})\gets (X_i,Y_i)$
    \State $C_{N_x},C_{N_y},C_{N_{xy}}\gets 1$
\Else
\State $C_N=\text{IncrementMorrisCounter}(C_N)$
    \If{$X_i=x_{\text{test}}$}
    \State $C_{N_x} \gets \text{IncrementMorrisCounter}(C_{N_x})$
    \If{$Y_i=y_{\text{test}}$}
    \State $C_{N_y} \gets \text{IncrementMorrisCounter}(C_{N_y})$
    \State $C_{N_{xy}} \gets \text{IncrementMorrisCounter}(C_{N_{xy}})$
    \EndIf 
 \ElsIf{$Y_i=y_{\text{test}}$}
    \State $C_{N_y} \gets \text{IncrementMorrisCounter}(C_{N_y})$ 
\EndIf 
\If{$\max\{C_{N_x},C_{N_y},C_{N_{xy}}<2M\}$}
\If{$C_N=M$}
\State $C_{\text{MI}}\gets C_{N_x}+C_{N_y}-C_{N_{xy}}$
    \State $\theta_{N_{xy}} \gets a-\frac{C_{\text{MI}}-(\mu+\eta)}{6M}$
    \State $s \gets$ \text{IncrementBiasEstimation}($S_{\text{bias}},s,\theta_{N_{xy}}$)  
    \State $C_N\gets 1$
    \EndIf
    \Else
    \State $C_N\gets 1$
\EndIf
\EndIf
\EndFor
\State $\hat{\theta}_{\text{MI}}\gets \frac{s-1}{S_{\text{bias}}-1}$
    \State $\hat{I}\gets 6M(\hat{\theta}_{\text{MI}}-a)$
\end{algorithmic}
\end{algorithm}
\begin{enumerate}
\item The algorithm receives an \textit{accuracy parameter} $\beta>0$ and an \textit{overhead parameter} $c>1$.
    \item In each iteration of the algorithm we collect a fresh pair of samples $(X,Y)\in [n]\times [m]$ according to $p_{XY}$, and store their values.  
    Assuming the received sample is $x$, we proceed to estimate $\log (p_xp_y/p_{xy})$ based on more fresh samples. 
    \item We use \textit{four} Morris counters - one that approximates a clock, one that approximates a count for $x$ values, one that approximates a count for $y$ values, and one that approximates a count for the pair $(x,y)$. The first of these counters have $M=B +1$ states, where $B$ is the is the smallest integer $k$ such that$\lceil (n\cdot m)^c \rceil\leq 2^k$. This counter (denoted as $C_N$) approximates a clock that counts until $\lceil (n\cdot m)^c \rceil$ samples from the distribution are observed. The second, third and fourth counters run in parallel to the first one and approximate a counter for $x$, a counter for $y$, and a counter for the pair $(x,y)$, and we denote their outputs as $C_{N_x}$, $C_{N_y}$ and $C_{N_{xy}}$, respectively. These counters each have $2M$ states, to guarantee they do not exceed the first counter with high probability. In the event that any of them reaches state $2M$ \textit{before} the first counter, we draw a fresh sample and initialize all counters. 
  \item We define $C_{\text{MI}}=C_{N_x}+C_{N_y}-C_{N_{xy}}$, and let $\overline{C}_{\text{MI}}=C_{\text{MI}}-\mu-\E\log N$ be the centralized version of $C_{\text{MI}}$. This is an almost unbiased estimator for $-\log (p_xp_y/p_{xy})$.   
    \item We now increment a bias estimation machine with $S_{\text{bias}}=\left\lceil\frac{36M^2}{\beta^2\delta}\right\rceil+1$ states whose purpose is to simulate the expectation operation.  
    Specifically, each time the first Morris counter concludes a count, we generate a $\mathop{\mathrm{Ber}}(\theta_{N_{xy}})$ random variable, with $\theta_{N_{xy}}=a-\frac{ \overline{C}_{\text{MI}}}{6M}$, and use it as the input to our bias estimation machine. The offset $a\triangleq \frac{4M-\E(\log N)-\mu}{6M}$ guarantees that $\theta_{N_{xy}}\in [0,1)$ with probability $1$, as $-2M\leq C_{\text{MI}}\leq 4 M$ since $C_{N_x},C_{N_y},C_{N_{xy}}$ are the outputs of Morris counters with $2M$ states. Our estimator for the mutual information $\hat{I}$ is the bias estimate of the machine, after subtraction of the known offset $a$ and multiplication by $6M$, that is, $\hat{I}=6M(\hat{\theta}-a)$.
\end{enumerate}

\subsection{Number of states of mutual information estimator}
$n,m,t,\beta,\delta,c,M,\eta,a,$ and $S_{\text{bias}}$ are program constants, so we do not count them in the memory consumption of the algorithm. At each time point, our algorithm keeps the value of a pair $(x,y)$, which requires $n\cdot m$ states, the state of the Morris counter approximating the clock, the state of the Morris counter approximating the logarithm of the $x$ counter, and the state of the bias estimation machine. Thus, the total number of states is the product of the individual number of states needed at each step, and recalling that $M=B+1\leq c\log nm+2$, the total number of states is
\begin{align}
    S\leq nm\cdot M\cdot (2M)^3 \cdot \left(\frac{36M^2}{\beta^2\delta}+2\right) = nm\left(\frac{288\cdot (c\log nm+2)^6}{\beta^2\delta}+16(c\log nm+2)^4\right).
\end{align}

\subsection{Analysis of the algorithm for $t=\infty$}
Let $(X,Y)$ be the fresh sample pair collected at the start of an algorithm iteration. 
We begin our analysis by showing that 
$\overline{C}_{\text{MI}}$ is close in expectation to $I(X;Y)$. Let $C_{N_X}^\infty,C_{N_Y}^\infty$ and $C_{N_{XY}}^\infty$ denote the corresponding infinite Morris counters. We first analyze the algorithm for these counters, and then appeal to Lemma~\ref{lem:Mor_bias} to bound the deviation of the limited memory counters used in our algorithm.
\begin{lemma}\label{lem:biasMIsmall}
 Denote $d_c(n,\alpha)=n^{-\alpha(c-1)+v_n(\alpha)}$ and let
 \begin{align}
 \phi_c(n,m)=2(e+1)\max\{d_c(n,1),d_c(m,1)\}+3\cdot\min\{1, C\cdot  \max\{d_c(n,1/2),d_c(m,1/2),d_c(nm,1/2)\}\}.
 \end{align}
We have
\begin{align}
  |\E (\overline{C}_{\text{MI}}^\infty) +I(X;Y)|\leq 3\cdot 10^{-5}+\phi_c(n,m).  
\end{align}   
\end{lemma} 
\begin{proof}
  Following the proof of Lemma~\ref{lem:biasmall}, we write
\begin{align}
    \E(\overline{C}_{\text{MI}}^\infty\mid X,N_X)&=\E(C_{N_X}^\infty+C_{N_Y}^\infty-C_{N_{XY}}^\infty-\mu -\E(\log N)\mid X,N_X)\\&= \log N_X+\log N_Y-\log N_{XY}-\E(\log N)+\gamma_{N_{X|Y}},
\end{align}
 where $\gamma_{N_{X|Y}}=\gamma_{N_X}+\gamma_{N_Y}+\gamma_{N_{XY}}$. We then have
 \begin{align}
   \E(\overline{C}_{\text{MI}}^\infty\mid X) &=\E\left(\log \frac{N_XN_Y/N_{XY}}{N}\mid X\right)+\E(\gamma_{N_{X|Y}}\mid X) \\&=\log\frac{p_Xp_Y}{p_{XY}}+\E\left(\log \frac{N_XN_Y/N_{XY}}{N\cdot p_Xp_Y/p_{XY}}\mid X\right) +\E(\gamma_{N_X}\mid X)
 \end{align}
 implying that $    \E(\overline{C}_{\text{MI}}^\infty)=-I(X;Y)+\E\left(\log \frac{N_XN_Y/N_{XY}}{N\cdot p_Xp_Y/p_{XY}}\right)+\E(\gamma_{N_{X|Y}})$.
Decomposing 
\begin{align}
  \E\left(\log \frac{N_XN_Y/N_{XY}}{N\cdot p_Xp_Y/p_{XY}}\right)= \E\left(\log \frac{N_X}{N\cdot p_X}\right)+\E\left(\log \frac{N_Y}{N\cdot p_Y}\right)-\E\left(\log \frac{N_{XY}}{N\cdot p_{XY}}\right) 
\end{align}
and applying Lemma~\ref{lem:smallratio} to each term separately, we have
\begin{align}
-d_c(nm,1)\leq \E\left(\log \frac{N_XN_Y/N_{XY}}{N\cdot p_Xp_Y/p_{XY}}\right)\leq 2(e+1)\max\{d_c(n,1),d_c(m,1)\}, 
\end{align}
and, similarly, recalling Lemma~\ref{lem:gammaB}, we have that
\begin{align}
   \E(\gamma_{N_{X|Y}})\leq 3\cdot 10^{-5}+3\cdot\min\{1, C\cdot  \max\{d_c(n,1/2),d_c(m,1/2),d_c(nm,1/2)\}\}.
\end{align}
This implies that in the counting phase of the algorithm we obtain an estimate for (minus) the mutual information that has an average bias bounded from above by
    \begin{align}
       3\cdot 10^{-5}+\phi_c(n,m)=3\cdot 10^{-5}+
       O\left(\min\left\{2^{\sqrt{\log n}}\cdot n^{-\frac{1}{2}\cdot (c-1)}, 2^{\sqrt{\log m}}\cdot m^{-\frac{1}{2}\cdot (c-1)}\right\}\right). 
    \end{align}
\end{proof}
The additive expected error resulting from the truncation of the Morris counters $C_{N_X}^\infty,C_{N_Y}^\infty$ and $C_{N_{XY}}^\infty$ at state $2M$ is upper bounded according to Lemma~\ref{lem:Mor_bias} and the triangle inequality by $\frac{ 300(c\log nm+2)}{(nm)^c(1-0.5(nm)^{-c})^2}$, which is asymptotically negligible w.r.t $\phi_c(n,m)$. In a similar fashion to Lemma~\ref{lem:bias_mach}, the input sequence to the bias estimation machine is an i.i.d. sequence with distribution $\mathsf{Bern}\left(\theta_{\text{MI}}\right)$, where
\begin{align}
     \theta_{\text{MI}}=
     \E (\theta_{N_{XY}}) = \E \left(a-\frac{ \overline{C}_{\text{MI}}}{6M}\right)=a+\frac{I(X;Y)+b_{\text{MI}}}{6M}.
\end{align}
As we set $S_{\text{bias}}=\left\lceil\frac{36M^2}{\beta^2\delta}\right\rceil+1$ and $\hat{I}=6M(\hat{\theta}_{\text{MI}}-a)$, we have
\begin{align}
    \E(\hat{I}-(I(X;Y)+b_{\text{MI}}))^2=36M^2\cdot \E(\hat{\theta}_{\text{MI}}-\theta_{\text{MI}})^2\leq \beta^2\delta,
\end{align}
and we obtain the $(\veps,\delta)$ guarantee from Chebyshev's inequality, i.e., $\Pr(|\hat{I}-(I(X;Y)+b_{\text{MI}})|>\beta)\leq\delta$.   
\subsection{Lower Bound}
For simplicity of proof, let $\veps,\delta\geq \frac{1}{300}$, and recall that $\veps<\frac{1}{12\ln 2}$. 
Our lower bound from Theorem~\ref{thm:lower_bound} implies that for joint entropy estimation of $H(X,Y)$ where $(X,Y)\in[n] \times [m]$, the memory complexity is $\Omega (n\cdot m)$. Assume that we have a mutual information estimation machine that returns an estimate of $I(X;Y)$ with additive error at most $\veps$ with probability at least $1-\delta$ using $S_{\text{MI}}^*(n,m,\veps,\delta)$ states. We show below an algorithm that uses this machine as a black box and estimates $H(X,Y)=H(X)+H(Y)-I(X;Y)$ with additive error of at most $3\veps$ with probability at least $1-3\delta$ using $S_{\text{MI}}^*\cdot O(\log^3 n\cdot \log^3 m)$ states. Since estimation of $H(X,Y)$ requires $S^*(n\cdot m,3\veps,3\delta)=\Omega(n\cdot m)$, this must imply that
\begin{align}
S_{\text{MI}}^*(n,m,\veps,\delta)>\Omega\left(\frac{n\cdot m}{\log^3 n\cdot \log^3 m}\right).
\end{align}
We now describe such an algorithm. The algorithm has $3$ modes. It starts in mode $1$, in which $H(X)$ is estimated. It then moves to mode $2$, in which $H(Y)$ is estimated, and finally it moves to mode $3$ in which $I(X;Y)$ is estimated. The current mode is stored using $S_1=3$ states. The estimation of each of the $3$ quantities above is done using
\begin{align}
  \tilde{S}=\max\{S^*(n,\veps,\delta) ,S^*(m,\veps,\delta),S_{\text{MI}}^*(n,m,\veps,\delta)\}  
\end{align}
states. Those states are ``reused'' once the algorithm switches its mode of operation. The algorithm is as follows:
\begin{enumerate}
    \item Start in Mode $1$.
    \item Increment a Morris counter with  $S_2=O(\log\log n)$ states at each observation of $X$. This counter determines the run time of mode $1$ and we denote it \textit{RunMode}$X$.
    \item Estimate $H(X)$ using the Morris-counter entropy estimator we introduced in Section~\ref{sec:upper} with $S^*(n,\veps,\delta)$ states.
    \item As RunMode$X$ arrives at state $S_2$, save the estimate $\hat{H}(X)$ of $H(X)$ using $S_3=O(\log^2 n)$ states.
    \item Switch to Mode $2$.
    \item Increment a Morris counter with  $S_4=O(\log\log m)$ states at each observation of $Y$. This counter determines the run time of mode $2$ and we denote it \textit{RunMode}$Y$.
    \item Estimate $H(Y)$ using the entropy estimator with $S^*(n,\veps,\delta)$ states.
    \item As RunMode$Y$ arrives at state $S_4$, save the estimate $\hat{H}(Y)$ of $H(Y)$ using $S_5=O(\log^2 m)$ states.
    \item Switch to Mode $3$.
    \item Estimate $I(X;Y)$ using the black-box machine with $S_{\text{MI}}^*(n,m,\veps,\delta)$ states.
    \item From this time onward, estimate $H(X,Y)$ as $\hat{H}(X)+\hat{H}(Y)-\hat{I}(X;Y)$, where $\hat{I}(X;Y)$ is the current estimate of the black-box machine.
\end{enumerate}
The idea here is that, after a long enough time, the entropy estimator output will be accurate enough, at which point we can store that value and switch modes. In order to decide if enough time has passed, we must ensure that the bias estimation machine, which outputs our entropy estimates, is sufficiently mixed. From Lemma~\ref{lem:mix}, we have that the mixing time of the bias estimation machine is at most $4S_{\text{Bias}}\log S_{\text{Bias}}\leq O(\log^3 n)$ samples, as $S_{\text{Bias}}=O(\log^2 n)$ states. Thus, it suffices to run the machine for $\log^k n$ samples of independent $\mathrm{Ber}(\theta)$ random variables for $k\gg 1$ and then stop it, which would guarantee it is sufficiently mixed. In order to save memory we use another Morris counter with $S_2=O(\log \log^k n)=O(\log\log n)$ states that determines when the mode run ends. We then store the state of the bias estimation machine, which corresponds to our estimate $\hat{H}(X)$ of $H(X)$, using $S_3=S_{\text{Bias}}=O(\log^2 n)$ states. At this point, the algorithm switches to mode $2$, and estimates $H(Y)$ with $S^*(m,\veps,\delta)$ states. As in mode $1$, we use a Morris counter of $S_4=O(\log\log m)$ states to determine when the machine is sufficiently mixed and can be stopped, and store the state of the bias estimation machine, which corresponds to the estimate $\hat{H}(Y)$ of $H(Y)$ this time, using $S_5=O(\log^2 m)$ states. The process then moves to state $3$ where $I(X;Y)$ is estimated using the black-box machine and, subsequently, the machine estimates $H(X,Y)$ as $\hat{H}(X)+\hat{H}(Y)-\hat{I}(X;Y)$, where $\hat{I}(X;Y)$ is the current estimate of the black box machine. All in all, this algorithm produces a $(3\veps,3\delta)$ (recall that we assumed $\delta,\veps\geq 1/100$) estimate of $H(X,Y)$ using 
\begin{align}
S\leq\tilde{S}\prod_{i=1}^5 S_i=\tilde{S}\cdot O(\log^3 n\cdot\log^3 m),
\end{align}
which implies that
\begin{align}
  \tilde{S}= \Omega\left(\frac{S}{\log^3 n \log^3 m}\right)= \Omega\left(\frac{S^*(n,m,3\veps,3\delta)}{\log^3 n \log^3 m}\right)=\Omega\left(\frac{n\cdot m}{\log^3 n \log^3 m}\right) . 
\end{align} 
Finally, since Theorem~\ref{thm:upper_bound} states that $S^*(n,\veps,\delta)=O(n\cdot\log^4 n)$ and $S^*(m,\veps,\delta)=O(m\cdot\log^4 m)$, and we assumed that $\frac{n}{\log^3 n}=\Omega(\log^7 m)$ and $\frac{m}{\log^3 m}=\Omega(\log^7 n)$, we must therefore have that 
\begin{align}
   S_{\text{MI}}^*(n,m,\veps,\delta)=\Omega\left(\frac{n\cdot m}{\log^3 n\cdot \log^3 m}\right). 
\end{align}

\section{Conclusions and open problems}
Due to the limitation $\veps>10^{-5}$, our upper bound is not informative when very small additive error is required. Indeed, the Morris counter seems to be inadequate in these regimes and, despite many follow up works, we are not aware of an improved analysis that cancels out the $10^{-5}$ term of~\cite{flajolet1985approximate}. A natural question to ask then is whether this is a true limitation arising as a result of the bounded memory or an artifact of the Morris counter. This gives rise to two potential directions for future research:
\begin{itemize}
\item Is there a counting algorithm with the same memory consumption as the Morris counter that does not suffer from this bias? 
\item Can we find an entropy estimator with similar number of states without this lower bound on the attainable additive error?
\end{itemize}
  Another interesting research direction is to close the $ \mathop{\mathrm{poly}}(\log n)$ gap between our upper and lower bounds w.r.t the dependence on $n$. It seems plausible to us that the upper bound is tight, .i.e., that the real dependence on $n$ is $n\mathop{\mathrm{poly}}(\log n)$, as $n$ is the minimal number of states needed to save one sample, and we must save our running entropy estimate as well. One possible reason for this mismatch between the bounds might be that our lower bound relies on reduction to the uniformity testing problem, which does not fully utilize the properties of a finite-state entropy estimator. In particular, the reduction is from estimation to binary hypothesis testing (testing uniform vs. $\veps$-far from uniform), whereas in $\veps$-additive entropy estimation we effectively have $\frac{\log n}{2\veps}$ hypotheses.  Particularly, it would seem that the binary test of $H(p)=\log n$ vs. $H(p)\leq \log n -\veps$ is easier as there is only one distribution with $H(p)=\log n$ (uniform). Hence, another preliminary approach for lower bounds might be 
  \begin{itemize}
      \item Solve the binary hypothesis testing problem  $H(p)=\alpha\log n$ vs. $H(p)\leq \alpha\log n -\veps$ for some $0<\alpha<1$. 
  \end{itemize}      As there are many distributions with entropy $\alpha\log n$, and since solving this problem immediately implies a lower bound on entropy estimation, this approach might help in improve upon our lower bound. 
\section*{Acknowledgements}
This work was supported by the ISF under Grants 1641/21 and 1766/22.

\bibliography{EntropyPaper}
\bibliographystyle{ieeetr}

\section*{Appendix}\label{sec:app}
\subsection{Proof of Lemma~\ref{lem:biasMSE}}
     Let $p_{i,j}$ be the transition probability from state $i$ to state $j$, and let $\pi_k$ be the unique stationary distribution of state $k\in[S]$. We first show that $\pi_k=\mu_k$, where $\mu_k$ is the $\mathrm{Binomial}(S-1,\theta)$ distribution, that is, 
     \begin{align}
         \mu_k = \binom{S-1}{k-1}p^{k-1}q^{S-k}.\label{eq:stat_dist}
     \end{align} 
     For brevity, denote $\mathsf{Bin}_p^S(k)=\binom{S-1}{k}p^{k}q^{S-k-1}$. As $\binom{S-1}{k-1}=\frac{k}{S-k}\binom{S-1}{k}$, we have $\mathsf{Bin}_p^S(k-1)=\mathsf{Bin}_p^S(k)\cdot \frac{k}{S-k}\cdot\frac{q}{p}$. Recall that if $\mu$ is the stationary distribution if and only if $\sum_{i=1}^S\mu_ip_{i,k+1}=\mu_{k+1}$ for any $k\in [S-1]$. Write
     \begin{align}
       \sum_{i=1}^S\mu_i p_{i,k+1}&=\mu_{k}p_{k,k+1}+  \mu_{k+1}p_{k+1,k+1}+ \mu_{k+2}p_{k+2,k} \\&=\mathsf{Bin}_p^S(k-1)\cdot \frac{S-k}{S-1}p+\mathsf{Bin}_p^S(k)\left(\frac{k}{S-1}p+\frac{S-(k+1)}{S-1}q\right)+\mathsf{Bin}_p^S(k+1)\cdot \frac{k+1}{S-1}q\\&=\mathsf{Bin}_p^S(k)\left(\frac{k}{S-k}\cdot\frac{S-k}{S-1}q+\frac{k}{S-1}p+\frac{S-(k+1)}{S-1}q+\frac{S-(k+1)}{k+1}\cdot \frac{k+1}{S-1}p\right)\\&=\mathsf{Bin}_p^S(k)(p+q)=\mu_{k+1}.
     \end{align}
 Now, due to the Ergodicity of the chain, when the machine is initiated with $\mathop{\mathrm{Bern}}(p)$ samples and run for a long enough time, eq.~\eqref{eq:stat_dist} implies that $M_t-1$ is distributed $\mathrm{Binomial}(S-1,p)$, thus the estimate $\hat{p}(M_t)$ has $\E(\hat{p}(M_t))=p$ and $\E(\hat{p}(M_t)-p)^2=\Var(\hat{p}(M_t))=\frac{pq}{S-1}\leq\frac{1}{S-1}$.
\subsection{Monte Carlo guarantee}
\begin{lemma}
    Monte Carlo simulation provides an $\alpha$-additive estimation for $\E (\log N)$ with probability $1-\delta$ with $L=\frac{(c\log n +4)^2+1}{\alpha^2\delta}$ samples.
\end{lemma}
\begin{proof}
Firstly, we have     
\begin{align}
        \E(N)=\sum_{k=1}^{M-1}\E(\tau_k)=\sum_{k=1}^{M-1} 2^k\leq2^{M},
    \end{align}
    and from Jensen's inequality $\E (\log N)\leq \log \E(N)\leq M$.
The proof follows from the Taylor series expansion of $\log N$.
\begin{align}
  \log N&= -\log \left(1-\frac{N-1}{N}\right)\\&=\sum_{k=1}^{\infty} \frac{1}{k}\left(\frac{N-1}{N}\right)^k\\&=\sum_{k=1}^{\E(N)-1} \frac{1}{k}\left(1-\frac{1}{N}\right)^k+\sum_{k=\E(N)}^{\infty} \frac{1}{k}\left(1-\frac{1}{N}\right)^k\\ &\leq \sum_{k=1}^{\E(N)-1} \frac{1}{k}+\frac{1}{\E(N)}\sum_{k=0}^{\infty}\left(1-\frac{1}{N}\right)^k\\ &\leq \log (E(N))+1+\frac{N}{E(N)}.
\end{align}
     Thus we can write
     \begin{align}
         \log^2 (N)\leq \log^2 (2E(N))+\frac{2N}{E(N)}\cdot \log (2E(N))+\frac{N^2}{E^2(N)},
     \end{align}
     and, taking expectation, we have 
     \begin{align}
         \E(\log^2 (N))&\leq \log^2 (2E(N)) +2\log (2E(N))+\frac{\Var(N)+\E^2(N)}{\E^2(N)}\\&\leq (M+1)^2+2(M+1)+2= (M+2)^2+1.
     \end{align}
Finally, from Chebyshev's inequality, we have
\begin{align}
    \Pr\left(\left|\frac{1}{L}\sum_{J=1}^L\log(N_j)-\E (\log N)\right|>\alpha\right)\leq \frac{\Var(\log N)}{L\alpha^2}\leq \frac{\E(\log^2 (N))}{L\alpha^2}\leq \frac{(c\log n +4)^2+1}{L\alpha^2},
\end{align}
thus taking $L=\frac{(c\log n +4)^2+1}{\alpha^2\delta}$ achieves the result.
\end{proof}
\subsection{Proof of Lemma~\ref{lem:Mor_bias}}
From total probability
\begin{align}
    \E(C_{N_X}^\infty)&=\Pr(C_{N_X}^\infty< 2M)\E(C_{N_X}^\infty\mid C_{N_X}^\infty< 2M)+\Pr(C_{N_X}^\infty\geq 2M)\E(C_{N_X}^\infty\mid C_{N_X}^\infty\geq 2M)\\&=\Pr(C_{N_X}^\infty< 2M)\E(C_{N_X})+\Pr(C_{N_X}^\infty\geq 2M)\E(C_{N_X}^\infty\mid C_{N_X}^\infty\geq 2M),
\end{align}
which implies
\begin{align}
   \E(C_{N_X}^\infty)-\Pr(C_{N_X}^\infty\geq 2M)\E(C_{N_X}^\infty\mid C_{N_X}^\infty\geq 2M)\leq\E(C_{N_X})\leq \frac{\E(C_{N_X}^\infty)}{\Pr(C_{N_X}^\infty< 2M)}.
\end{align}
Note that
\begin{align}
 \Pr(C_{N_X}^\infty=k\mid C_{N_X}^\infty\geq 2M)=\frac{\Pr(C_{N_X}^\infty=k)}{\Pr(C_{N_X}^\infty\geq 2M)}   
\end{align}
for $k\geq 2M$ and zero otherwise. Hence, 
\begin{align}
  \Pr(C_{N_X}^\infty\geq  2M)\E(C_{N_X}^\infty\mid C_{N_X}^\infty\geq  2M)=\sum_{k=2M}^\infty k\Pr(C_{N_X}^\infty=k).  
\end{align}
We thus have the upper bound
\begin{align}
 |\E(C_{N_X}) -\E(C_{N_X}^\infty)|\leq \max \left\{\sum_{k=2M}^\infty k\Pr(C_{N_X}^\infty=k),\hspace{1mm}\frac{\E(C_{N_X}^\infty)\Pr(C_{N_X}^\infty\geq 2M)}{\Pr(C_{N_X}^\infty< 2M)}\right\}.  
\end{align}
We first bound the second term. Write 
\begin{align}
  \E(C_{N_X}^\infty) &= \E(\log N_X) +\mu+\E(\gamma_{N_X})\label{eq:moThm}\\&\leq \log(\E(N_X))+\mu +1+10^{-5}\label{eq:lembound}\\&=\log(\E(N\cdot p_X))+\mu +1+10^{-5}  \\&<c\log n+3\label{eq:mubound},
\end{align}
where~\eqref{eq:moThm} follows from Theorem~\ref{thm:morris} and the smoothing theorem,~\eqref{eq:lembound} follows from Lemma~\ref{lem:gammaB} and Jensen's inequality, and~\eqref{eq:mubound} follows as $\mu \approx -0.3$ and from $\E(N)=\sum_{k=1}^{M-1} 2^k=2^{M}-2\leq 4n^c$. Now, note that if $C_{N_X}^\infty\geq 2M$ after $N=m$ samples, then the second counter must have moved from state $2M-1$ to state $2M$ in less than $m$ steps, implying that
\begin{align}
    \Pr(C_{N_X}^\infty\geq 2M\mid N=m)\leq m\cdot 2^{-(2M-1)}.  
\end{align}
Since $2^{-M}\leq 1/(2n^c)$, we have
\begin{align}
  \Pr(C_{N_X}^\infty\geq 2M)= \E(\Pr(C_{N_X}^\infty\geq 2M\mid N))\leq \E(N\cdot 2^{-(2M-1)})\leq 2n^{-c}. 
\end{align}
Combining the above, we get
\begin{align}
    \frac{\E(C_{N_X}^\infty)\Pr(C_{N_X}^\infty\geq 2M)}{\Pr(C_{N_X}^\infty< 2M)}\leq \frac{2(c\log n +3)}{n^c(1-2n^{-c})}.
\end{align}
We now proceed to carefully bound $\sum_{k=2M}^\infty k\Pr(C_{N_X}^\infty=k)$. For any $X=x$ we have
\begin{align}
  \sum_{k=2M}^\infty k\Pr(C_{N_x}^\infty=k)&=\sum_{m=1}^\infty \Pr(N=m)\sum_{n_x=1}^m\Pr(N_x=n_x\mid N=m) \sum_{k=2M}^\infty\Pr(C_{N_x}^\infty=k\mid N=m,N_x=n_x) \\&=\sum_{m=1}^\infty \Pr(N=m)\sum_{n_x=1}^m\Pr(N_x=n_x\mid N=m) \sum_{k=2M}^\infty\Pr(C_{n_x}^\infty=k\mid N=m).
\end{align}
We divide the computation into sample-state blocks, where each sample block is of length $\lceil 4n^c \rceil$ and each state block is of length $M$. Clearly $\Pr(N\in [\ell\cdot \lceil 4n^c \rceil,(\ell+1)\cdot \lceil 4n^c \rceil)\leq \Pr(N \geq \ell\cdot \lceil 4n^c \rceil)$, and $k\leq \alpha M$ in the interval $[(\alpha-1)M,\alpha M)$. Thus,
\begin{align}
 & \sum_{k=2M}^\infty k\Pr(C_{N_x}^\infty=k)\\\leq & \sum_{m=1}^\infty \Pr(N=m)\sum_{n_x=1}^m\Pr(N_x=n_x\mid N=m) \sum_{\alpha=2}^\infty(\alpha+1)M\max_{\alpha M \leq k\leq (\alpha+1)M}\Pr(C_{n_x}^\infty=k\mid N=m)\\\leq &\sum_{\ell=0}^\infty \Pr(N\geq \ell\cdot \lceil 4n^c \rceil) \sum_{\alpha=2}^\infty(\alpha+1)M\max_{\substack{\ell\cdot \lceil 4n^c \rceil\leq m\leq (\ell+1)\cdot \lceil 4n^c \rceil\\ 1\leq n_x\leq m\\ \alpha M \leq k\leq (\alpha+1)M}}\Pr(C_{n_x}^\infty=k\mid N=m).
\end{align}
First, we have from Lemma~\ref{lem:probN} that $\Pr(N\geq \ell\cdot \lceil 4n^c \rceil)\leq 5e^{-\ell}$. Now, note that if $C_{N_x}^\infty=k$, then the second counter must have moved from state $k-1$ to state $k$ in less than $n_x$ steps. Thus we have
\begin{align}
  \max_{\substack{\ell\cdot \lceil 4n^c \rceil\leq m\leq (\ell+1)\cdot \lceil 4n^c \rceil\\ 1\leq n_x\leq m\\ \alpha M \leq k\leq (\alpha+1)M}}\Pr(C_{N_x}^\infty=k\mid N=m)&\leq \max_{\substack{\ell\cdot \lceil 4n^c \rceil\leq m\leq (\ell+1)\cdot \lceil 4n^c \rceil\\ 1\leq n_x\leq m\\ \alpha M \leq k\leq (\alpha+1)M}}n_x\cdot 2^{-(k-1)}\\&\leq \max_{\substack{\ell\cdot \lceil 4n^c \rceil\leq m\leq (\ell+1)\cdot \lceil 4n^c \rceil\\ \alpha M \leq k\leq (\alpha+1)M}}m\cdot 2^{-(k-1)}\\&\leq (\ell+1)\cdot \lceil 4n^c \rceil \cdot 2^{-(\alpha M -1)} .
\end{align}
Plugging back the above, we have
\begin{align}
  \sum_{k=2M}^\infty k\Pr(C_{N_x}^\infty=k)&\leq 10M\cdot \lceil 4n^c \rceil  \left(\sum_{\ell=0}^\infty (\ell+1)e^{-\ell}\right) \left(\sum_{\alpha=2}^\infty (\alpha+1)2^{-\alpha M}\right)\\&\leq 10M\cdot \lceil 4n^c \rceil\cdot \frac{1}{(1-e^{-1})^2}\cdot\frac{3\cdot 2^{-2M}}{(1-2^{-M})^2}\\&< \frac{ 100(c\log n+2)}{n^c(1-0.5n^{-c})^2},
\end{align}
where we used the identity 
\begin{align}
  \sum_{n=N_1}^\infty n q^{n-1}=\frac{N_1q^{N_1-1}-(N_1-1)q^{N_1}}{(1-q)^2}< \frac{N_1q^{N_1-1}}{(1-q)^2}.  
\end{align}
thus, overall, 
\begin{align}
 |\E(C_{N_X}) -\E(C_{N_X}^\infty)|\leq \frac{ 100(c\log n+2)}{n^c(1-0.5n^{-c})^2}.
\end{align}


\end{document}